\documentclass[12pt]{article}
\usepackage{scicite}
\usepackage[bookmarks=false]{hyperref}
\usepackage{times}
\usepackage{bbm}
\usepackage{amsthm}
\usepackage{amssymb,amsfonts}
\usepackage{graphbox}
\usepackage{graphicx}
\usepackage{float}
\usepackage{extarrows}
\usepackage{tikz}
\usepackage{xcolor}
\usepackage{subcaption}

\newcommand{\E}{\mathbb{E}}

\newcommand{\dis}{\displaystyle}

\usepackage{booktabs} % For formal tables

\usepackage{color}

\newtheorem{theorem}{Theorem}

\newtheorem{corollary}{Corollary}
\newtheorem{proposition}{Proposition}

\newtheorem{definition}{Definition}

\newenvironment{sciabstract}{%
\begin{quote} \bf}
{\end{quote}}

\title{Predictability limit of partially observed systems}

\author
{Andr\'es Abeliuk,$^{1}$ Zhishen Huang,$^{2}$ Emilio Ferrara,$^{1\ast}$ Kristina Lerman,$^{1\ast}$\\
\\
\normalsize{$^{1}$Information Sciences Institute, University of Southern California, Marina del Rey, CA, 90292, USA}\\
\normalsize{$^{2}$University of Colorado Boulder, Boulder, CO, 80302, USA}\\
\\
\normalsize{$^\ast$To whom correspondence should be addressed; E-mail: emiliofe@usc.edu,  lerman@isi.edu}
}

\date{}

\begin{document} 

% Double-space the manuscript.

\baselineskip24pt

% Make the title.

\maketitle

% Place your abstract within the special {sciabstract} environment.

\begin{sciabstract}
Applications from finance to epidemiology and cyber-security require accurate forecasts of dynamic phenomena, which are often only partially observed. 
We demonstrate that a system's predictability degrades as a function of temporal sampling, regardless of the adopted forecasting model. We quantify the loss of predictability due to sampling, and show that it cannot be recovered by using external signals. We validate the generality of  our theoretical findings in  real-world partially observed systems representing infectious disease outbreaks, online discussions, and software development projects. 
On a variety of prediction tasks---forecasting new infections, the popularity of topics in online discussions, or interest in cryptocurrency projects---predictability irrecoverably decays as a function of sampling, unveiling fundamental predictability limits in partially observed systems.
\end{sciabstract}

% In setting up this template for *Science* papers, we've used both
% the \section* command and the \paragraph* command for topical
% divisions.  Which you use will of course depend on the type of paper
% you're writing.  Review Articles tend to have displayed headings, for
% which \section* is more appropriate; Research Articles, when they have
% formal topical divisions at all, tend to signal them with bold text
% that runs into the paragraph, for which \paragraph* is the right
% choice.  Either way, use the asterisk (*) modifier, as shown, to
% suppress numbering.
\newpage
\section*{Introduction}

Forecasting complex dynamic phenomena --- from influenza outbreaks to public opinions, stock market, and cyberattacks --- is central to many policy and national security applications~\cite{Vespignani09}. Prediction is also the standard framework in evaluating models of complex systems learned from data~\cite{Hofman2017}. Time series forecasting, which underpins popular models of dynamic phenomena,  represents a process as a sequence of observations (discrete or continuous numbers of events) at regular time intervals. After learning parameters from past observations, the models can be used to predict future observations \cite{chatfield2000time}. Forecasting models based on  stochastic and self-exciting point processes, autoregressive and hidden Markov models, have been developed to predict crime~\cite{short2008statistical,mohler2011self}, social unrest~\cite{EMBERS}, terrorism~\cite{raghavan2013hidden}, epidemics~\cite{scarpino2019predictability}, human mobility~\cite{Song1018Science}, personal correspondence~\cite{malmgren2009universality}, online activity~\cite{hogg2012social,stoddard2015popularity}, dynamics of ecological systems~\cite{garland2014model,garland2018anomaly} and more~\cite{sapankevych2009time}. 
%{\color{red} TO ADD MORE EXAMPLES: Flash crashes, more terrorism, positive behaviors.}

A fundamental challenge to modeling efforts is the fact that complex systems are seldom fully observed. For example, when estimating opinions in a social system, it is not practical nor feasible to interview every individual in the population; instead, polling is used to elicit responses from a representative sample of a population. When social media is used as a proxy of opinions, it is similarly impractical to collect all relevant posts; instead, a (pseudo-random) sample (e.g., the Twitter \textit{Decahose}), is often used. Further biases can emerge when data is deliberately manipulated or deleted, e.g., to obfuscate or censor content or activity~\cite{king2014reverse}. In short, the data used to learn predictive models of complex phenomena often represents a highly filtered and incomplete view. 

How does data loss due to sampling affect the predictability of complex systems and the accuracy of models learned from the data? 
Statisticians have developed a number of approaches to compensate for data loss, including data imputation~\cite{little2019statistical} to fill in missing values and evaluating the representativeness of the sample~\cite{morstatter2013sample,ruths2014social}. Few of these approaches apply to temporal data. To quantify the predictability of dynamic systems, researchers use measures such as autocorrelation and permutation entropy. The former measures similarity between a time series and its own time-lagged versions. Recently, permutation entropy was introduced as a model-free, nonlinear indicator of the complexity of data~\cite{bandt2002entropy,weighted_Perm_Entropy}.  Permutation entropy represents the complexity of a time series through statistics of its ordered sub-sequences, also known as motifs, and has been adopted to model predictability of ecological systems~\cite{pennekamp2018intrinsic,garland2014model} and epidemic outbreaks~\cite{scarpino2019predictability}. However, the impact of sampling on these measures of predictability of complex systems is not known. 
% Statisticians have approached this challenge as a problem of missing data and developed a range of methods for addressing it from data imputation~\cite{little2019statistical} to efficient estimation of distributional properties~\cite{valiant2013estimating}. In contrast, social scientists have worried about  the representativeness  of the sample~\cite{morstatter2013sample,ruths2014social}, or how robustly the data  predicts social phenomena~\cite{Song1018Science,scarpino2019predictability}. Previous work did not address two fundamental questions considered in this paper: %what are the limits of predictability arising from incompletely observed or sampled data.
% whether accurate models can be learned from incomplete observations, and how data loss due to sampling or filtering affects predictability of complex systems. %dynamic phenomena.

As the first step towards addressing this question, we model incomplete observation as a stochastic sampling process that selects events at random with some probability $p$ and drops the remaining events from observations of a system. This allows us to theoretically characterize how %sparse temporal sampling 
sampling decreases the autocorrelation of a time series. %measure the impact of sparse temporal sampling on the predictability of a process. 
%We show that sampling destroys useful information about the dynamic process, as demonstrated by a decrease in its autocorrelation %of its time series 
%and an increase in permutation entropy, a measure of a process's complexity~\cite{Bandt2002permutation}. Moreover, we demonstrate that it is impossible to recover lost information by using some external data, even if it is highly correlated with the original dynamic process.  
We then empirically show
that sampling also reduces the predictability of a dynamic process according to both autocorrelation and permutation entropy. Moreover, the loss of predictability cannot be fully recovered from some external signal, even using data highly correlated with the original unsampled process.
%We use model-independent metrics to quantify the predictability of a system in terms of %the decay in auto-covariance of the signal and the cross-covariance between the signal and an informative external signal. 
%We show that (1) autocorrelation of the signal decays with decreasing sampling rate; (2) the correlation of the sampled signal with the external signal also decays with the sampling rate, albeit slower than the signal's autocorrelation is lost. In addition, we show that (3) prediction error becomes non-stationary in time, presenting additional challenges to forecasting tasks. 
As a result, %of information loss, 
forecasts made by autoregressive models may be no better than predictions of simpler, less accurate models that assume independent events.
We validate these findings with both synthetic and real-world data representing complex social and techno-social systems. %from epidemiology and social media. 
Without any modeling assumptions on the data, we show how sampling systematically degrades the predictability of these systems.

Researchers increasingly predict complex systems and social network dynamics \cite{Vespignani09,rand2011dynamic,sekara2016fundamental,Hofman2017} to learn the principles of human and machine behavior~\cite{lazer2009computational,rahwan2019machine}. Practitioners and lawmakers alike often base their decisions on such insights~\cite{athey2017beyond,watts2017should}, including for public health~\cite{blumenstock2015predicting,pananos2017critical} and public policy \cite{johnson2016new,deville2016scaling,bail2018exposure,scheufele2019science}. 
As some pointed out~\cite{lazer2014parable,shiffrin2016drawing}, however,  caution should be used when drawing conclusions from incomplete data. %sampling---even random sampling as we show here---can alter the observed dynamic process, reducing its predictability. 
Sampling, even random sampling, %introduces correlated noise into the process, altering its 
distorts the observed dynamics of a process, reducing its predictability.
We formalize and quantify this common, yet understudied, source of bias %when studying dynamic phenomena
in partially observed systems. %Since observational data by definition involves sampling, our work uncovers fundamental limits of predictability in complex dynamic systems.

\section*{Results}
%%%%%%%%%%%%%%%%%%%%%%%%%%%%%%%%%%%%%% MODEL %%%%%%%%%%%%%%%%%%%%%%%%%%%%%%%%% 
\subsection*{Model}
Consider a dynamic process generating events, for example, social media posts mentioning a particular topic, or newly infected individuals during an epidemic. We can represent the process as a time series of event counts, $X=[X_1,X_2,\ldots,X_T]$, each entry representing the number of observations of $X$ at time $t$.
% of new social media posts on a topic or the monthly number of newly infected individuals. https://www.overleaf.com/project/5d314d7fb2c2b61551c23c59
We refer to this time series as the \textit{ground truth signal}.

Observers of this process may not see all events. Twitter, for example, makes only a small fraction ($\leq10\%$) of messages posted on its platform programmatically available. Similarly, hospitals may delay reporting new cases of a disease or  under-count them altogether when, for various reasons, people do not seek medical help after getting sick. We refer to the time series of observed events $Y=[Y_1,Y_2,\ldots,Y_T]$ as the \textit{observed signal}. Intuitively, $Y$ represents a sample of events present in the ground truth signal $X$.

\begin{figure}[tb]
\centering
\includegraphics[trim={0 1.5cm 0 1.5cm}, width=0.95\columnwidth]{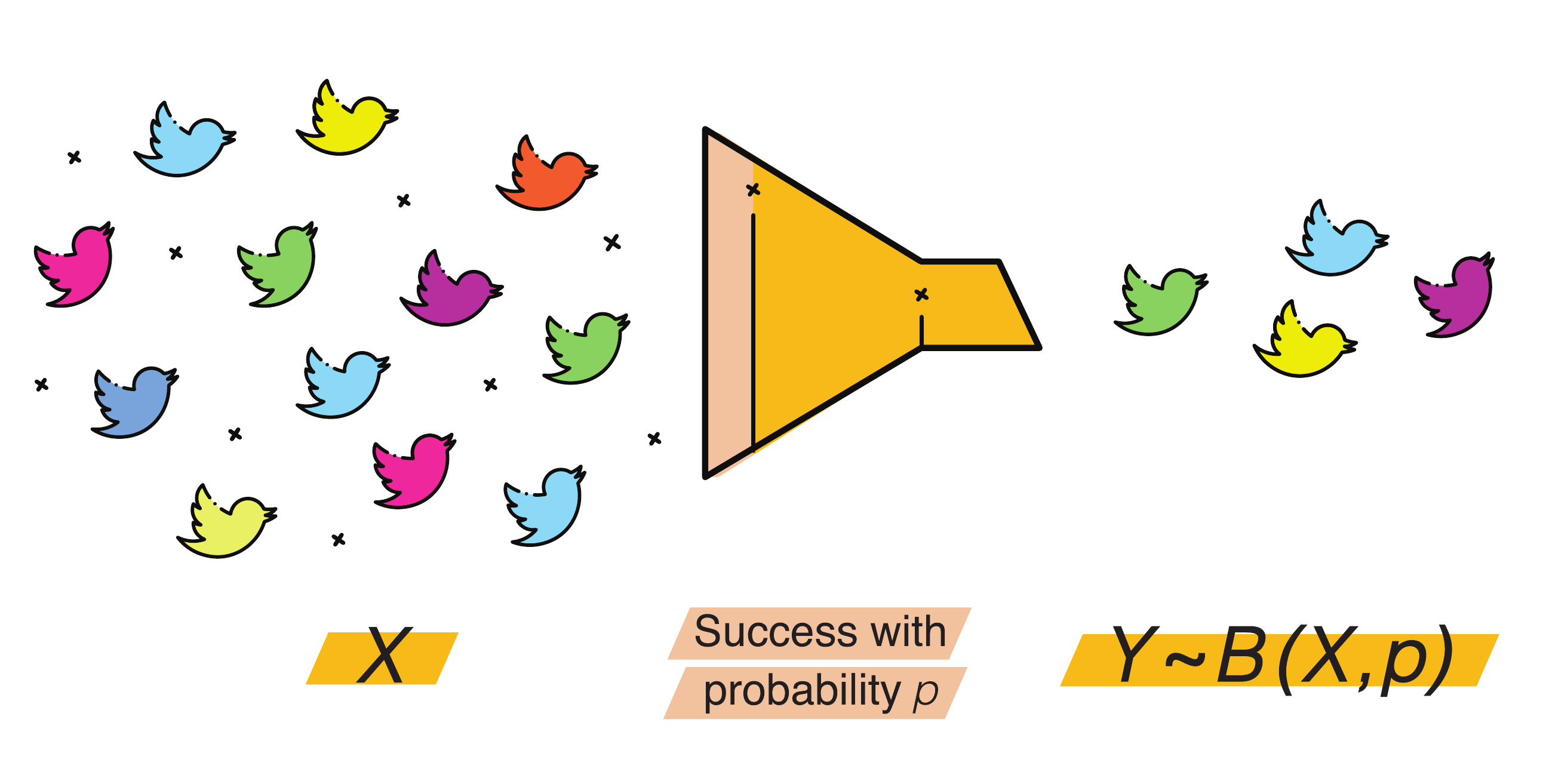}
\caption{
Sampling paradigm as a representation of a partially  observed dynamic process. Here,
%Sketch of the filtering framework. In the context of a social process,
$X$, the ground truth signal, represents the actual events, e.g., Twitter
posts mentioning a particular topic; %or individuals newly infected during an epidemic;
$Y$ represents the observed subset of events. The funnel illustrates the sampling process. % by the researcher.
The probability of an event being observed is $p$. The Binomial distribution $B(X,p)$ is used to model the observed signal $Y$.
}
% \caption{Sketch of the filtering framework. In the context of cyber attacks, $X$  is the number of attempted attacks per day;  $Y$ is the number of successful attacks per day observed by the target.  The probability of an attack being successful is $p$. The Binomial distribution $B(n,p)$ is used to model the time series $Y$. }
\label{fig:sketch}
\end{figure}

We model partial observation as a stochastic sampling process, where each event has some probability to be observed, independent of other events. This allows us to formalize how the time series of the ground truth and the observed signals are related. Figure~\ref{fig:sketch} illustrates this paradigm. 
\begin{definition}
We define sampling rate $p \in[0,1]$ as the percentage of events that are preserved by the observation process.
Let $X$ and $Y$ be two time series related by
$$Y \sim B([X],p), $$
where $B([X],p)$ is a Binomial process with %parameters $[X]$ for sample size and 
$[X]$ trials, each with success probability $p$.
\end{definition}

% External signals
The factors driving the system may also produce some external events that may help predict the %the ground truth signal. 
observed system. For example, rising temperatures associated with climate change may help better forecast epidemics that are made more virulent by changes in climate. Similarly, news reports may be associated with increased social media posts on specific topics, since both are driven by world events. Temperatures and news reports may provide important signals for predicting future events.
% Events are often driven by some external factors. For example, the spread of an epidemic may be boosted by the rising temperature in the breakout region, and %yet the price of Bitcoin may make some hacking tools cheaper to acquire, increasing the likelihood of cyber attacks.
% news reports may drive social media posts on specific topics.
%The time series of these external factors may be informative for predicting the ground truth signal $Y$.
\begin{definition}
% We refer to the time series of the informative %external factors 
% events as the \textit{external signal},  defined as $S=S_1,S_2,\ldots,S_T$.
We define the \textit{external signal} as the time series $S=S_1,S_2,\ldots,S_T$ that may provide information about the ground truth signal.
\end{definition}

\subsection*{Quantifying the Loss of Predictability}
Researchers have devised measures of predictability of complex systems. At the simplest level, autocorrelation captures how well a time series representing a complex system is correlated with its own time-lagged versions. This indicator of predictability is popular in finance~\cite{Lim2013stockindPrediction}.
In ecology and physics, permutation entropy %$H^\mathrm{P}_{(\textrm{w})}$
is used to measure predictability~\cite{Bandt2002permutation,garland2014model}. Permutation entropy (PE) captures the complexity of a time series through statistics of its ordered sub-sequences, or motifs  (see \textit{Materials \& Methods}).  %using ordinal analysis. The complexity can be interpreted as the entropy of the trends among the sub-sequences with certain length. Therefore the 
The higher the permutation entropy, the more diverse the motifs, which in turn renders the time series less predictable.
Permutation entropy was shown to be strongly related to Kolmogorov-Sinai (KS) entropy~\cite{politi2017quantifying},  a theoretical measure quantifying the complexity of a dynamical system.  KS is not easy to reliably estimate from data; however, for one dimensional time-series, KS and permutation entropy are known to be equivalent under a variety of conditions~\cite{bandt2002entropy}.
%  %The ordinal pattern means the relative magnitude relation among successive time series values. As an example, if $x_1= 3, x_2 = 6, x_3 = 1$, then the ordinal pattern of this subsequence $\{x_1,x_2,x_3\}$ is $\phi(x_1,x_2,x_3) = (312)$ because $x_3\le x_1\le x_2$(see Methods for more details).
Using different forecasting models, Garland et al.~\cite{garland2014model} demonstrated an empirical correlation between predictability of the models and %weighted 
permutation entropy~\cite{weighted_Perm_Entropy}.
Since then, PE has been used as a model-free indicator of predictability of infectious disease outbreaks~\cite{scarpino2019predictability}, human mobility~\cite{Song1018Science}, ecological systems~\cite{pennekamp2018intrinsic}, and anomaly detection in paleoclimate records~\cite{garland2018anomaly}.
Besides autocorrelation and PE, we also use prediction error as a measure of predictability~\cite{garland2014model}. 
%\textcolor{red}{AA: the following sentence seems misplaced} 
However, since prediction error depends on the forecasting model, we explore it in detail only with synthetic data %time series generated by an auto-regressive process 
(\textit{SI}, \textit{Synthetic Data Experiments}).
% \S\ref{appendix::Synthetic_Data}).

We show that sampling reduces predictability of a signal, and the %lower the sampling rate, 
more data is filtered out, the less predictable the signal becomes. The loss of predictability cannot be recovered using an informative external signal, even if it is highly correlated with the original ground truth signal. We  develop a framework for quantifying predictability loss due to sampling and validate it empirically using all measures of predictability.
%\note{KL: Add text about correlated noise.}{}

%\subsection*{Theoretical Results}
Our main theoretical contribution is an analytical characterization of the covariance matrix of the observed signal $Y$ in terms of the ground truth signal $X$ and the sampling rate $p$ (cf., \textit{Materials \& Methods}, Theorem~\ref{prop::covariance_dropout_quadratic}).  
Theorems and their proofs are presented in the \textit{SI}.
Based on this characterization, we derive two results stating the effects of sampling on the predictability of the observed signal $Y$:

\begin{description}

\item[Decay of autocorrelation of the observed signal.]
The autocorrelation (defined as Pearson correlation between values of the signal at different times) of the observed signal $Y$ decays monotonically at lower sampling rates (Corollary~\ref{coro::autocorrelation}, \textit{Materials \& Methods}).

\item[Decay of covariance with the external signal.]
The correlation between the observed and external signals degrades linearly at lower sampling rates (Corollary~\ref{coro::covariance}, \textit{Materials \& Methods}).
\end{description}

Specifically, to quantify the impact of sampling on the predictability of a signal, we first derive the autocorrelation of the observed signal as a function of the sampling rate $p$ (cf., Corollary~\ref{coro::auto_corr}, \textit{Materials \& Methods}).   When $p=1$ (i.e., complete observation), we recover the autocorrelation of the ground truth signal $X$. At lower sampling rates, the autocorrelation decays as postulated above. 
In parallel, we demonstrate empirically that sampling degrades predictability as measured using permutation entropy. 

A forecasting model may compensate for the loss of predictability by leveraging an informative external signal. For example, auto-regressive forecasting models allow for additional covariates to improve predictions~\cite{box1975intervention}.
However, according to our second result, %the information lost to sampling 
predictability cannot be fully recovered with an external signal, even one that is highly correlated with the ground truth signal. 

%We demonstrate the loss of predictability in several real-world and synthetic time series data, showing that the autocorrelation of the observed signal decreases at lower sampling rates, and its permutation entropy increases. In addition, the co-variance and mutual information between the external and the observed  signals both decrease at lower sampling rates. 

%%%%%%%%%%%%%%%%%%%%%%%%%%%%%%%% EMPIRICAL .  %%%%%%%%%%%%%%%%%%%%%%

\subsection*{Empirical Results}
\label{sec:empirical}
We show that sampling irreversibly degrades the predictability of real-world complex systems, studying three phenomena: disease outbreaks, online discussions, and software collaborations.  Sampling reduces predictability according to both autocorrelation and permutation entropy measures, and the observed decay of autocorrelation agrees with theoretical predictions. 
% explore the effects of partial observability using our framework on real-world data.
% The numerical experiments on synthetic data generated by an ARIMA process demonstrate that the prediction performance decay quadratically with information loss rate.
% We find that autocorrelation  decreases with sampling rate, in line with theoretical predictions, as does covariance with an informative  external signal. 
%Real-world data is seldom that clean and ideal, yet, even from a model-free perspective, information loss affects predictability of real-world temporal data from a variety of domains, such as disease outbreaks; social media; and popularity of open source projects.

Predictability cannot be fully recovered using an informative external signal.
In addition to co-variance, we use \textit{mutual information} (MI) to measure the shared information between the external and the observed signals~\cite{leung1990information}.  
Mutual information quantifies the reduction in uncertainty about one random variable due to the presence of another~\cite{delsole2004predictability}, and like PE it captures the non-linearities in the data that covariance cannot measure. We empirically find  that sampling reduces both the covariance and MI with the external signal.

\paragraph{Epidemics.}
 Scarpino \& Petri~\cite{scarpino2019predictability} used permutation entropy to show that predictability of disease outbreaks decreases over longer time periods, suggesting changes in the behavior of epidemics over time.
%Here, using data from eight diseases (Chlamydia, Gonorrhea, Hepatitis A, Influenza, Measles, Mumps, Polio, and Whooping cough), 
Here, we show that %partial observation 
the predictability of epidemics is also affected by how partially or fully observed the new infections are.

%Here, using data from eight diseases (Chlamydia, Gonorrhea, Hepatitis A, Influenza, Measles, Mumps, Polio, and Whooping cough),  We find that all diseases show a clear decrease in predictability as lower sampling rates reduce the number of observed infections. % measured using two measures: weighted permutation entropy and Pearson's autocorrelation.

We study eight diseases (Chlamydia, Gonorrhea, Hepatitis A, Influenza, Measles, Mumps, Polio, and Whooping cough), representing each disease outbreak as a time series of the weekly number of reported infections in each US state. 
% \note{KL: Do you average over all states? its the median across all states, which is described on the caption of Figure 2} 
We find that at lower sampling rates, the permutation entropy (PE), over one-year moving windows (although the results are robust to longer windows, see \textit{SI} Figure~\ref{fig:exp_epidemics_w104}), of the times series increases (Figure~\ref{fig:exp_epidemics} (top-left)) and the autocorrelation decreases (Figure~\ref{fig:exp_epidemics} (top-right)). %consistent with a  decrease in predictability of the outbreak. 
Given that each disease has a different base PE and autocorrelation coefficient (see \textit{SI}, Figures~\ref{fig:epidemics_theo_si} and \ref{fig:exp_epidemics_si} for the absolute values), we normalized the predictability measure of the sampled time series  by the corresponding measure of the %unfiltered
%original time series 
ground truth time series (i.e., with full information, corresponding to sampling rate $p=1$) to capture the relative change. 
% To validate our theoretical claims, we compare the empirical autocorrelation obtained for each %state-level
% disease at different sampling rates against theoretical predictions (Equation~\ref{eq:autocorr}).  
The observed loss of autocorrelation for each disease outbreak at different sampling rates (Figure~\ref{fig:exp_epidemics} (bottom)) agrees well with the theoretical predictions derived by Equation~\ref{eq:autocorr}. Our findings suggest that %the temporal sample used to measure an outbreak introduces enough noise to bias the temporal dynamics of a disease.
observing only a subset of the new infections distorts the observed dynamics of the disease, making the outbreak less predictable. 

\begin{figure}[H]
\begin{center}
\includegraphics[width=0.49\linewidth]{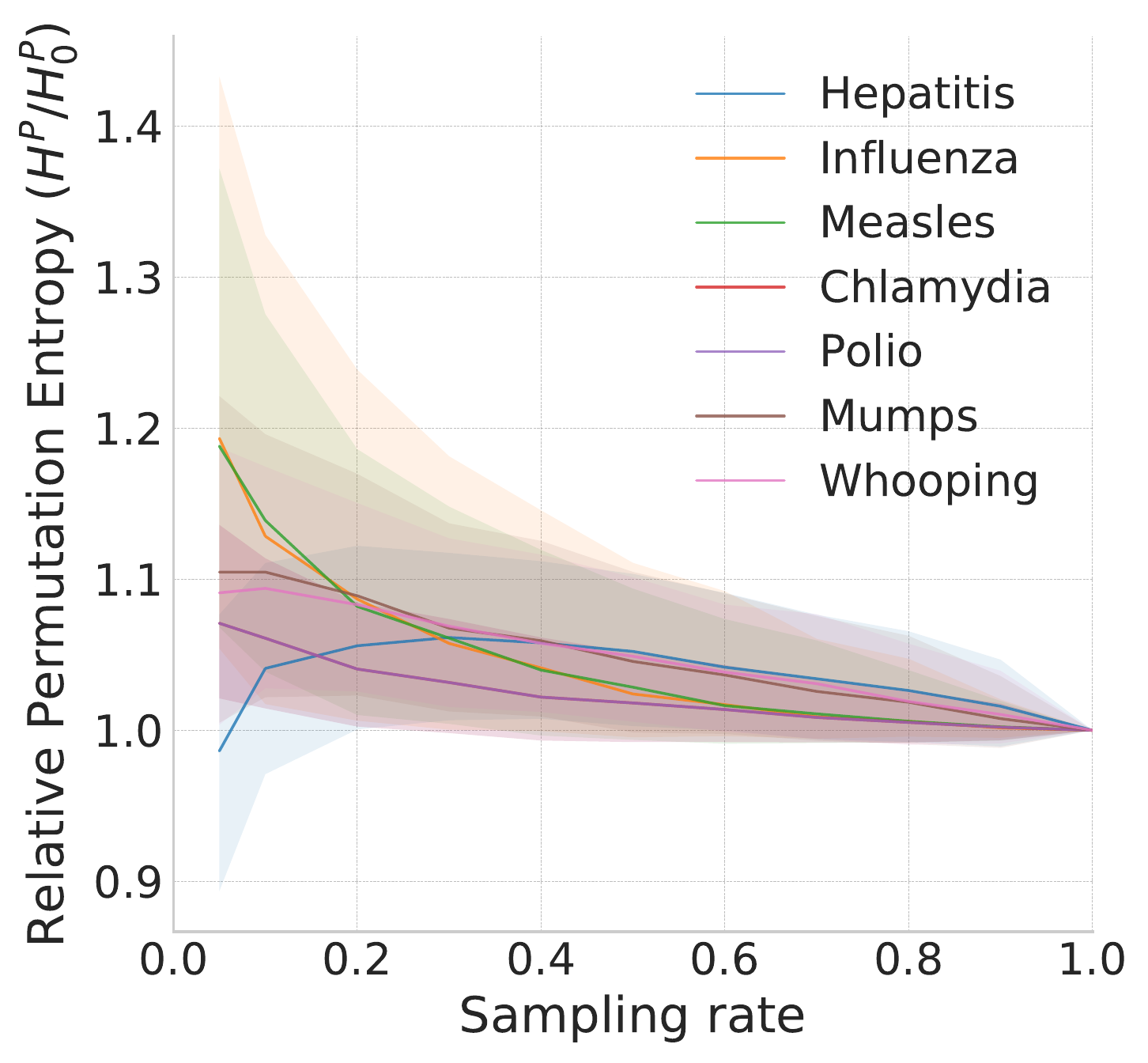}
% \hspace{0.1cm}
\includegraphics[width=0.49\linewidth]{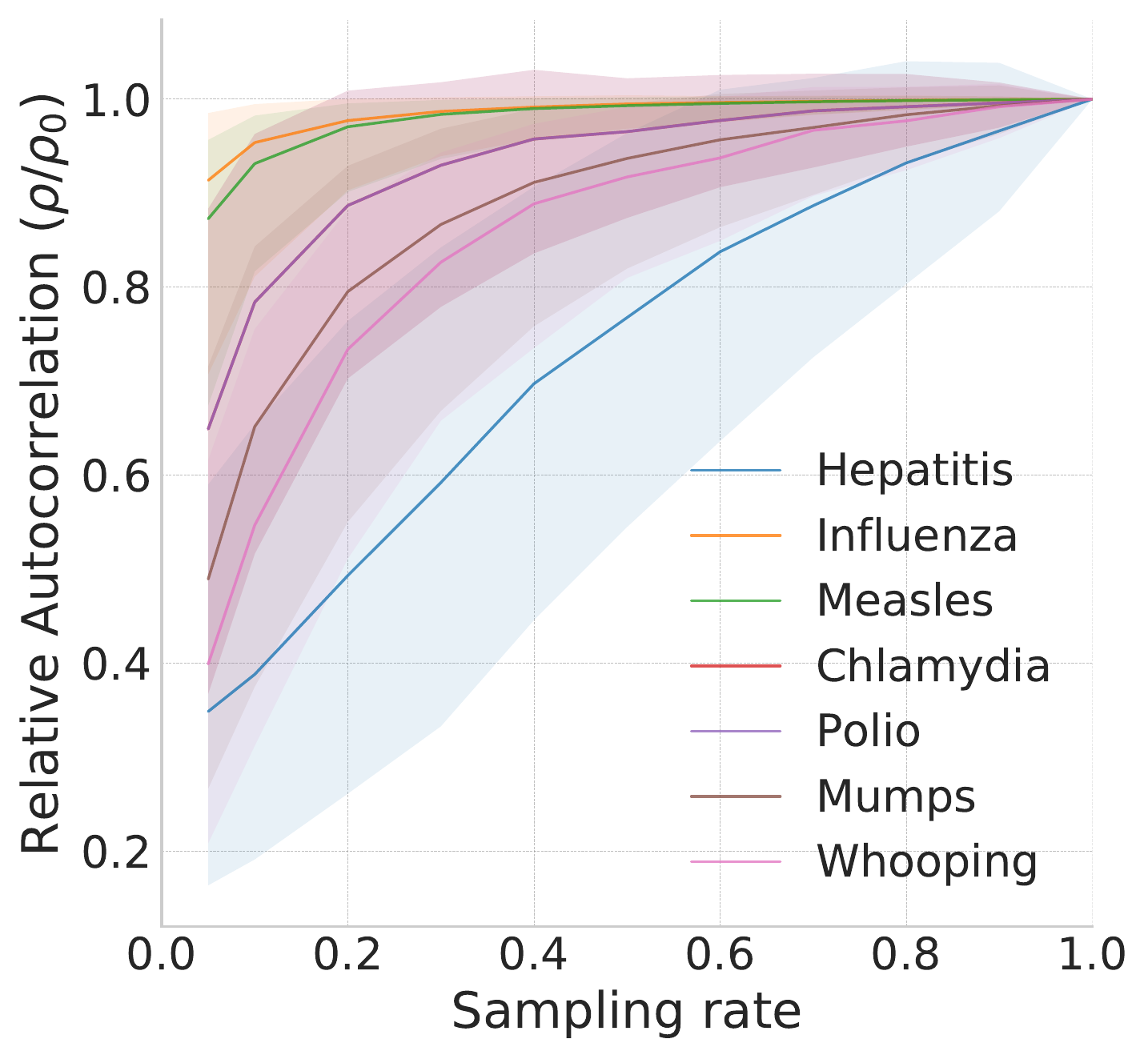} \\
% \hspace{0.1cm}
\includegraphics[width=0.49\linewidth]{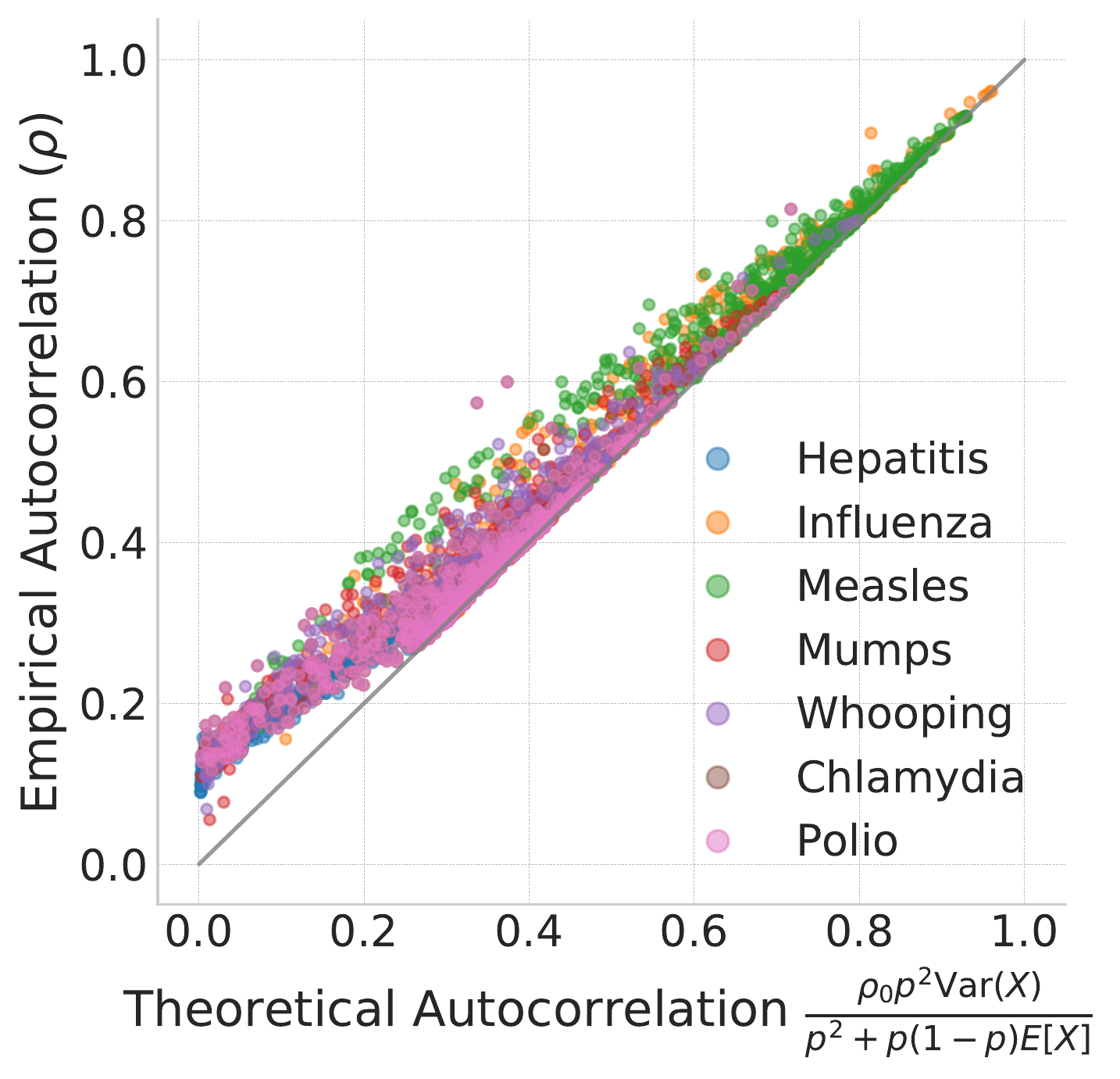}
\caption{Loss of predictability of disease outbreaks due to sampling.   The plots show a decrease in permutation entropy (\textbf{top-left}) and an increase in autocorrelation (\textbf{top-right}) of the outbreak time series for increasing sampling rates. For each of the eight %weekly, state-level 
diseases, we selected $100$ random one-year time windows and calculated the relative weighted permutation entropy  and autocorrelation for different sampling rates over that window. The solid line represents the median ratio across all states between the original time series and the sampled one; shaded regions mark the inter-quartile ranges.  The \textbf{bottom} plot supports our theoretical results by plotting Equation~\ref{eq:autocorr}  against the empirical autocorrelation of the sampled time series at different sampling rates for each disease. } %state-level disease.}
\label{fig:exp_epidemics}
\end{center}
\end{figure}

\begin{figure}[t]
\begin{center}
\centering
\includegraphics[width=0.49\columnwidth]{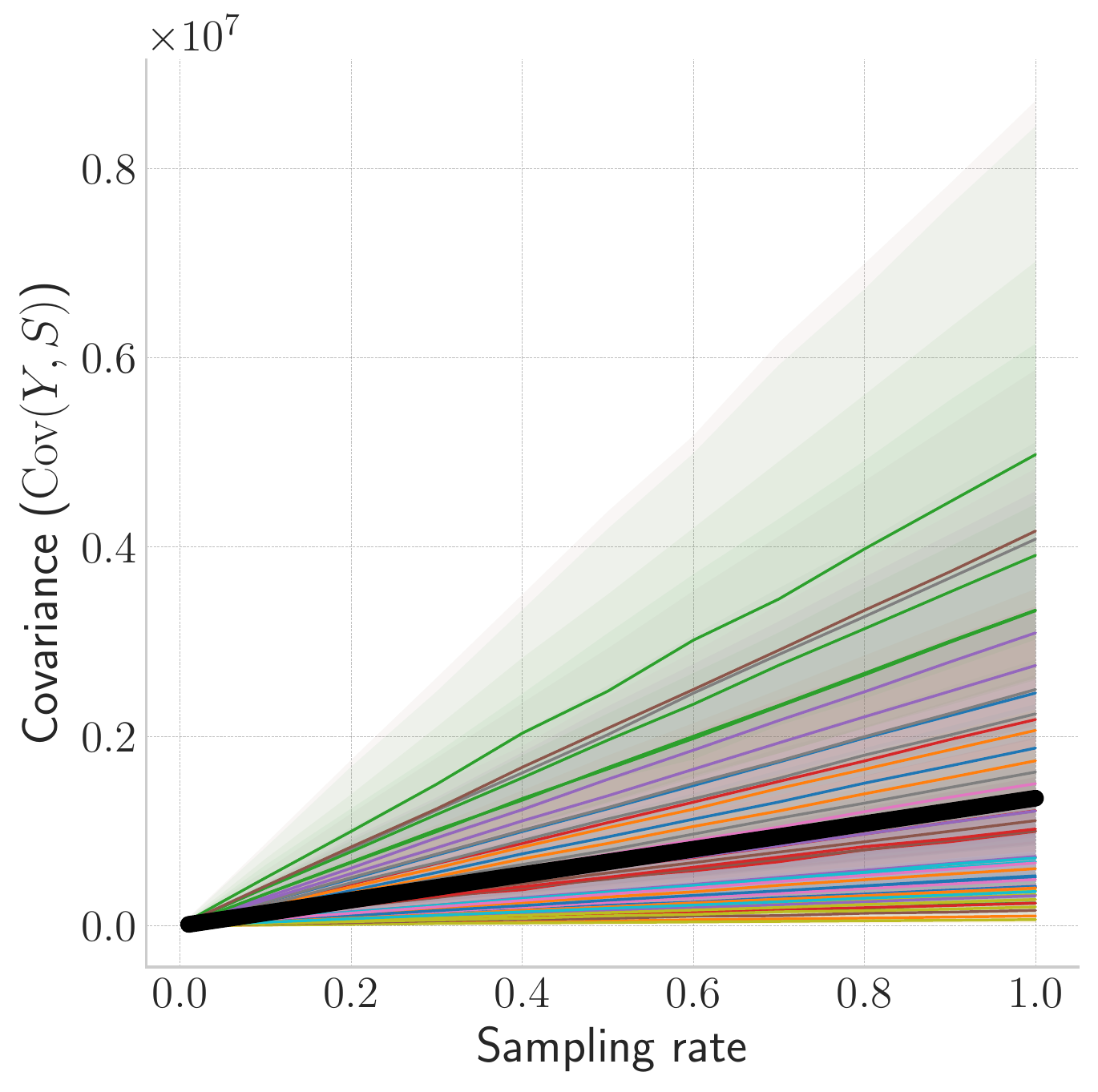}
%\hspace{0.3cm}
\includegraphics[width=0.49\columnwidth]{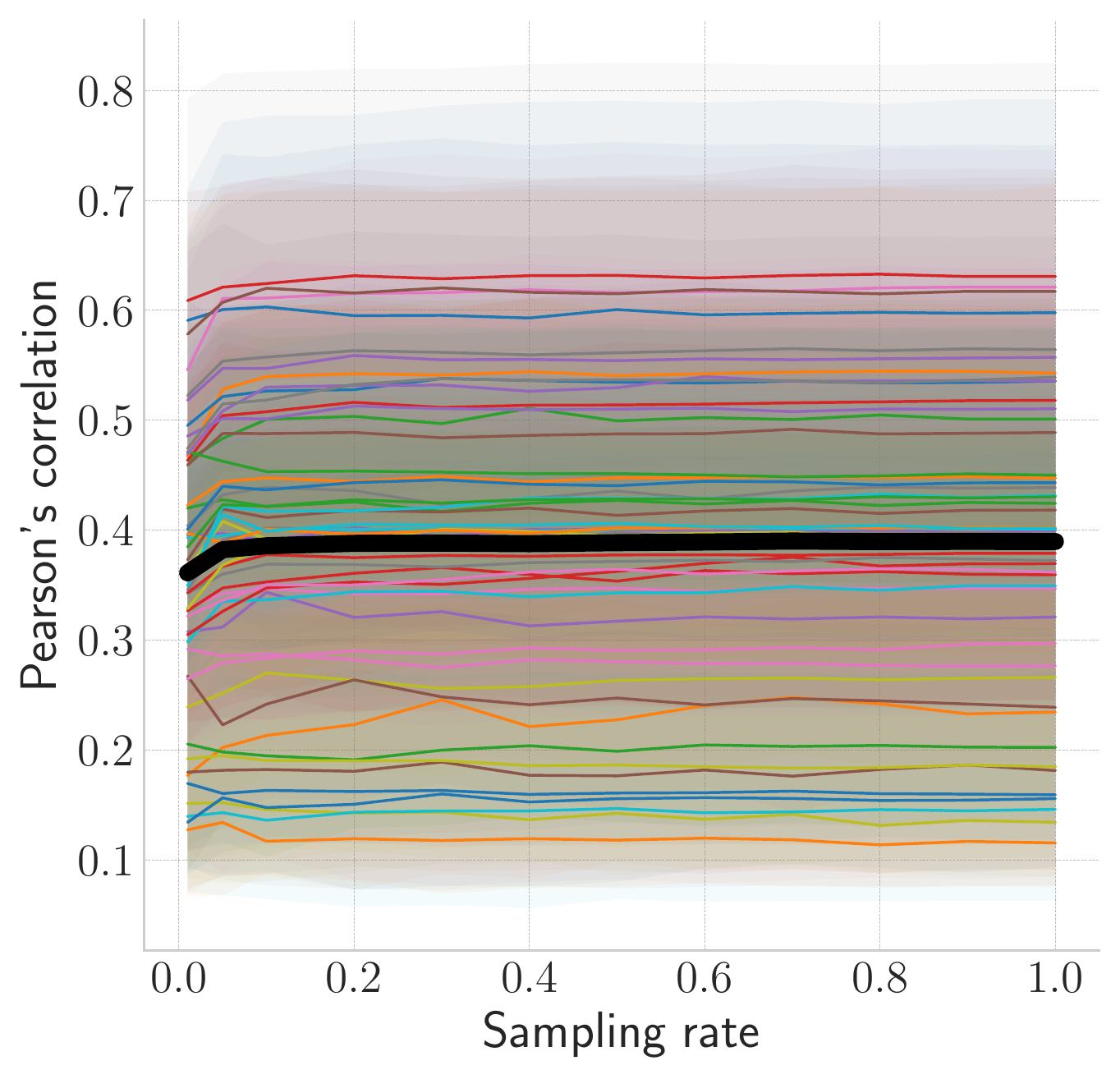}
\caption{Decay of covariance between ground truth and external signals. For each state, we selected $100$ random one-year time windows and calculated the median covariance (\textbf{left}) and Pearson's correlation (\textbf{right}) between Google Flu trends and the influenza activity at different sampling rates. Shaded regions mark the inter-quartile ranges for each state; the solid line represents the average coefficient across all states.}
\label{fig:exp_epidemics_ext_signal}
\end{center}
\end{figure}

Next, we use influenza data to validate Corollary~\ref{coro::covariance}, which states that an external signal becomes less informative (i.e., has lower covariance) about the ground truth data at lower sampling rates. %about %claims a linear decay of the covariance with a predictive external signal. 
As an external signal $S$, we use state-level Google Flu trends~\cite{ginsberg2009detecting}, which estimate influenza activity based on search queries. Figure~\ref{fig:exp_epidemics_ext_signal} (left) shows a linear growth of covariance for each state's influenza time series with increasing sampling rate. However, as depicted on the right plot, there is no observed loss of correlation for lower sampling rates. This is due to the large variance relative to the mean exhibited by influenza activity.  From Theorem~\ref{prop::covariance_dropout_quadratic}, we have that the standard deviation of the observed signal $Y$ is  $$\sigma_Y=\sqrt{p^2\mathrm{Var}(X) + p(1-p)\E[X]}\approx p\; \sigma_X$$ when $\mathrm{Var}(X)\gg\E[X]$.
Then, it follows from Corollary~\ref{coro::covariance}  and  
the definition of Pearson's correlation $\rho$, that 
$$\rho_{Y,S}=\frac{\mathrm{Cov}(Y,S)}{\sigma_Y \sigma_S} \approx \frac{p\; \mathrm{Cov}(X,S)}{ p\; \sigma_X \sigma_S}=\rho_{X,S}.$$
Thus, the linear decrease of covariance is offset by a linear decrease of the standard deviation. However, this is not always the case, as we later show with the cryptocurrency popularity scenario.

% \note{KL: Add mutual information decay for epidemics and external signal.}{}

% Besides co-variance, we use \textit{mutual information} (MI) to measure the shared information between the two time series~\cite{leung1990information}.  
% Mutual information captures the non-linear dependencies in the data that covariance cannot measure, and it %quantifies the average predictability between two variables
% captures the reduction in uncertainty of one random variable due to the presence of another~\cite{delsole2004predictability}.
Supplementary Figure~\ref{fig:exp_epidemics_mi} shows that mutual information between Google Flu Trends and influenza activity also decreases, suggesting that the former becomes less informative about influenza activity the more it is sampled. 

\paragraph{Social Media.}
Next we consider the problem of predicting social media activity. We analyze the popularity of hashtags on Twitter, defined as the daily number of posts using that hashtag. We focus on the 100 most frequently used hashtags in our  data (cf., \textit{Material \& Methods}), and for each hashtag, we sample from all posts mentioning the hashtag several times at different rates to produce multiple sampled time series.  %We then compute the autocorrelation and permutation entropy of each time series.

%\begin{center}
\begin{figure}[H]
\centering
\includegraphics[align=c,width=0.49\linewidth]{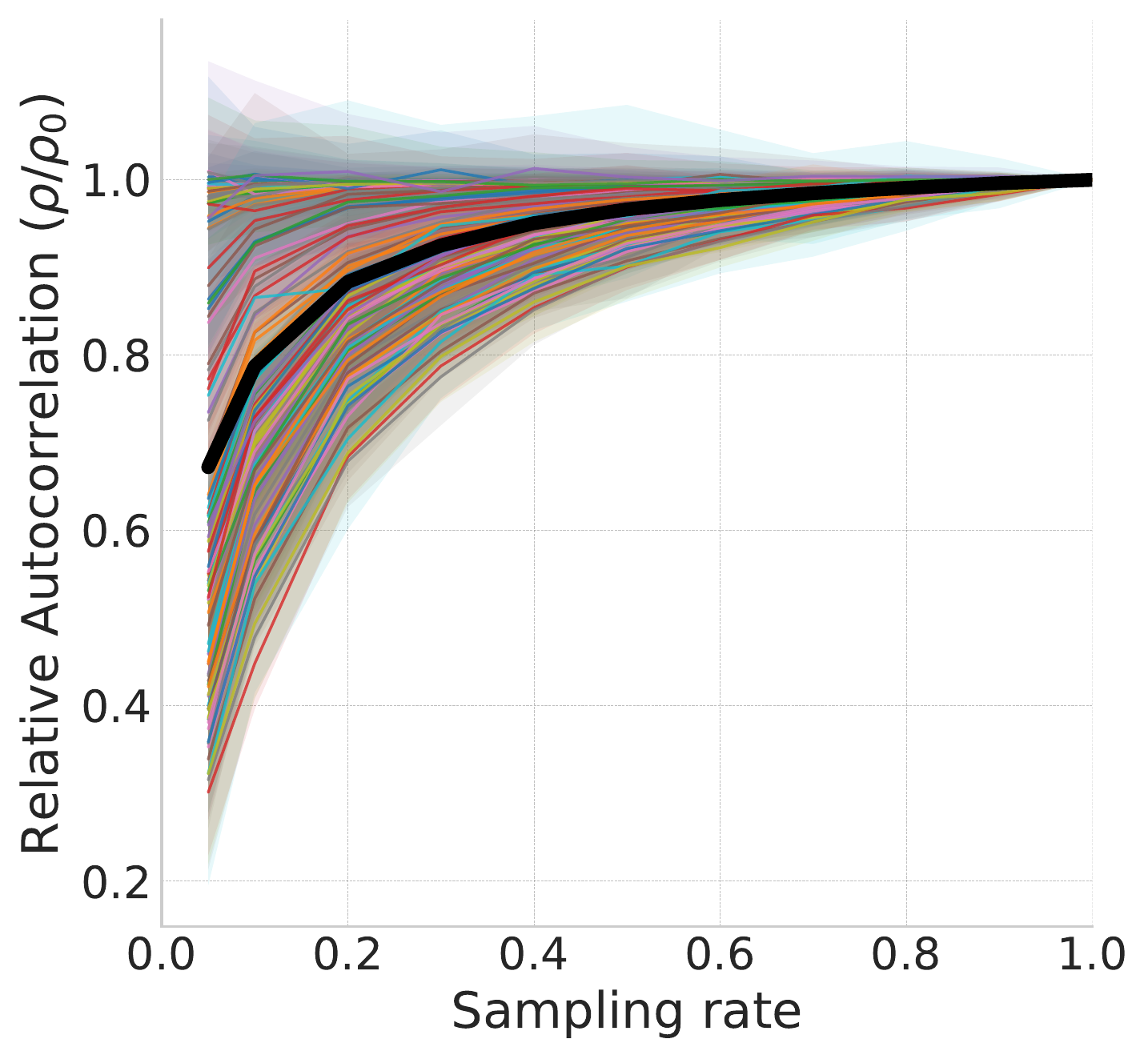}
%\hspace{0.4cm}
\includegraphics[align=c,width=0.49\linewidth]{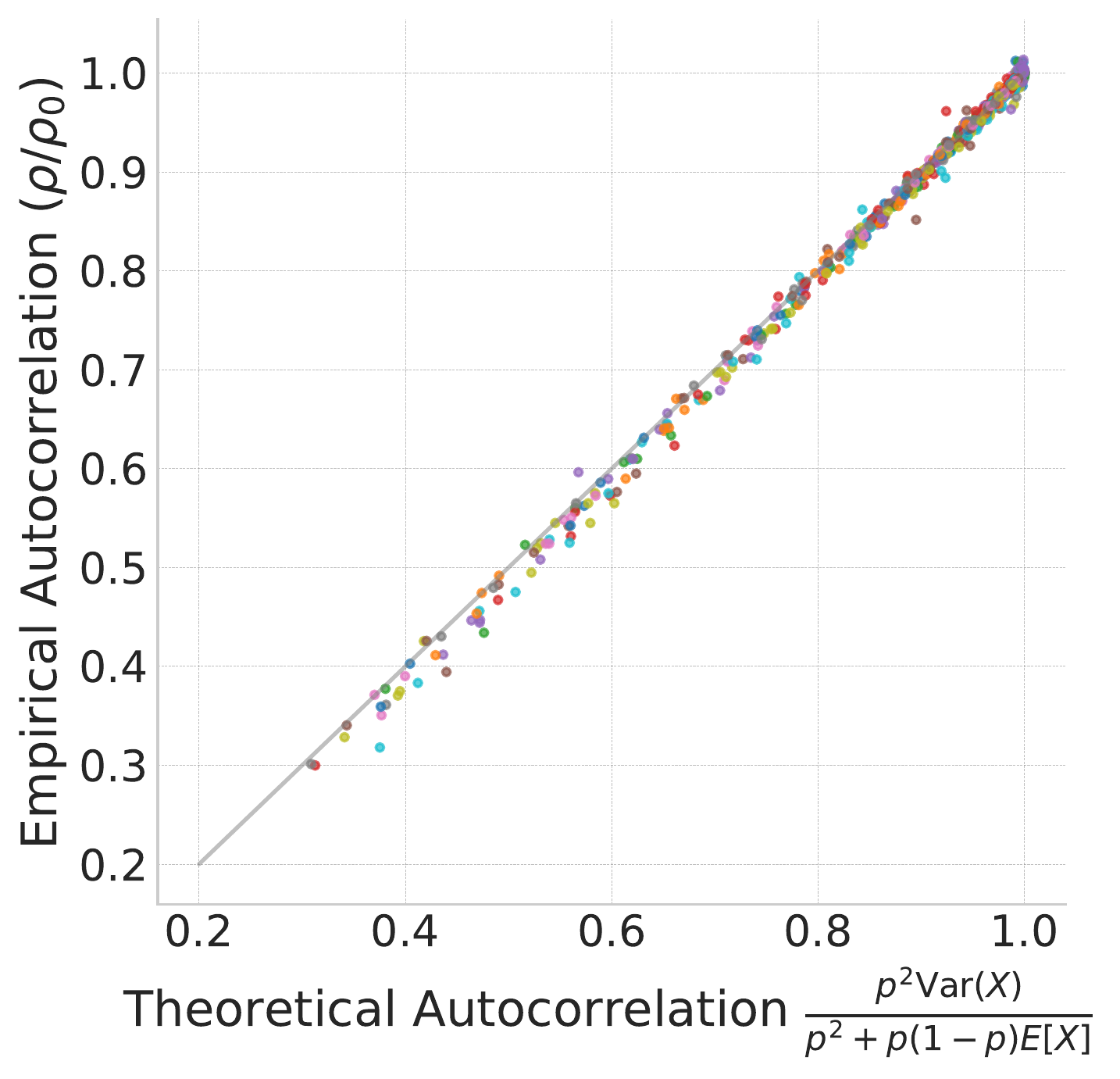}
\caption{Empirical and theoretical effects of sampling on autocorrelation of hashtag popularity. (\textbf{left}) Median autocorrelation relative to the original time series for 100 most popular hashtags; shaded regions mark the inter-quartile ranges; the black line represents the average autocorrelation across all hashtags.
(\textbf{right}) Accuracy of the theoretical prediction according to Equation \ref{eq:autocorr}. 
% An identity line (solid grey) is drawn to represent the perfect fit of the data.
} 
\label{fig:exp_hashtag}
\end{figure}
%\end{center}

Figure~\ref{fig:exp_hashtag} (left) shows the effects of sampling at different rates on the autocorrelation of hashtags' popularity. The plot shows the median autocorrelation loss relative to the original time series. For each ground truth signal, {we found the most significant autocorrelation time lag, which is kept fixed during the down sampling process to calculate autocorrelation at different sampling rates;} then, we plotted the median ratio between the original and sampled autocorrelation.
% Hence, for instance, a value of around $0.7$ for dropout rate $0.8$ means that on average, at that rate, the autocorrelation of the sample data is $70\%$ that of the original data. Autocorrelation of the sampled signal drops precipitously around dropout rate of 80\%.
Although the curvatures are different for each hashtag, all time series are accurately characterized by our theoretical results (Equation~\ref{eq:autocorr}): Figure~\ref{fig:exp_hashtag} (right) shows that the empirical loss of autocorrelation fits the theoretical predictions. Figure~\ref{fig:exp_user} (\textit{SI}) reports the results for the sampled time series of Twitter user activity, measured by the daily number of user's posts.

%To evaluate the impact of dropout on prediction, we used the same sampling strategy, but for each sampled time series we computed the weighted permutation entropy.
The loss of predictability is also seen when using  permutation entropy with the same sampling strategy.
Figures~\ref{fig:exp_twitter_entropy1} and \ref{fig:exp_twitter_entropy2} (\textit{SI}) show a clear trend in entropy increase  (i.e., decrease of predictability) for both user activity and popularity of hashtags. The loss of predictability for user activity, for instance, happens in 63\% of the users, while the rest of the cases comprise of time-series whose PE mostly do not change, except for low sampling rates (see Figure \ref{fig:exp_twitter_entropy_cluster} (\textit{SI})).

% The right plot of Figure \ref{fig:exp1} depicts the average ratio between the RMSE normalized by the mean of the original data and the sampled one. Thus, an increase of the relative normalized RMSE means that the % error of the prediction has increases relative to the error in the original time series.
% prediction error for the filtered time series increases relative to the prediction error of the original time series. Similar to the synthetic experiments, prediction error of the filtered signal increases with the dropout rate.

Note that, in many applications, researchers use data from the Twitter \textit{Decahose} or the \textit{streaming} API, which capture approximately 10\% and 1\% sample of tweets, i.e., sampling rates of 0.1 and 0.01 respectively~\cite{morstatter2013sample}. Considering that, at such low sampling rates, %prediction errors increase almost by a factor of two \note{AA: i dont know where this factor of two came from}, compared to using 
relative autocorrelation may be half of its value using the complete Twitter stream (\textit{Firehose}), care should be taken when drawing conclusions from the partially observed system.

\paragraph{Cryptocurrency Popularity.}
%In this section, we validate our findings regarding the loss of correlation between a sampled time series and an informative external signal, than can be used as a covariate.
We present additional findings regarding loss of correlation between a sampled time series and an external signal.
We study the effect of the price of cryptocurrencies on the adoption of said technology by software developers. To measure interest in the technology behind a cryptocurrency, we track the popularity of Github projects whose description is associated to that cryptocurrency. The four most popular cryptocurrencies during the collection period spanning January 2015 to March 2015 were Bitcoin (BTC), Litecoin (LTC), Monero (XMR) and Ripple (XRP). Some cryptocurencies, like Ethereum, were also popular, but since they were not yet publicly launched, we excluded them from the following analysis.

% As an example, Figure~\ref{fig:exp4} (top) shows the daily number of new ``watchers'' that Bitcoin Core integration project has received, along with the daily closing price of Bitcoin (bottom). Both time series are highly correlated, with Pearson's coefficient of $0.87$.
% % Moreover, as a side note, Granger causality tests identify that price affects watches, but the opposite direction does not hold.

% \begin{center}
% \begin{figure}[t]
% \centering
% \includegraphics[width=0.8\columnwidth]{Figures/bitcoin.pdf}
%         \caption{Bitcoin popularity on Github. (Top) The daily number of watches received by the Bitcoin repository. (Bottom) The daily closing price of Bitcoin cybercurrency. The time frame of the data is January 1st, 2015 to August 31st, 2017.  }
% \label{fig:exp4}
% \end{figure}
% \end{center}

%\begin{center}
\begin{figure}[t]
\centering
\includegraphics[align=c,width=0.49\columnwidth]{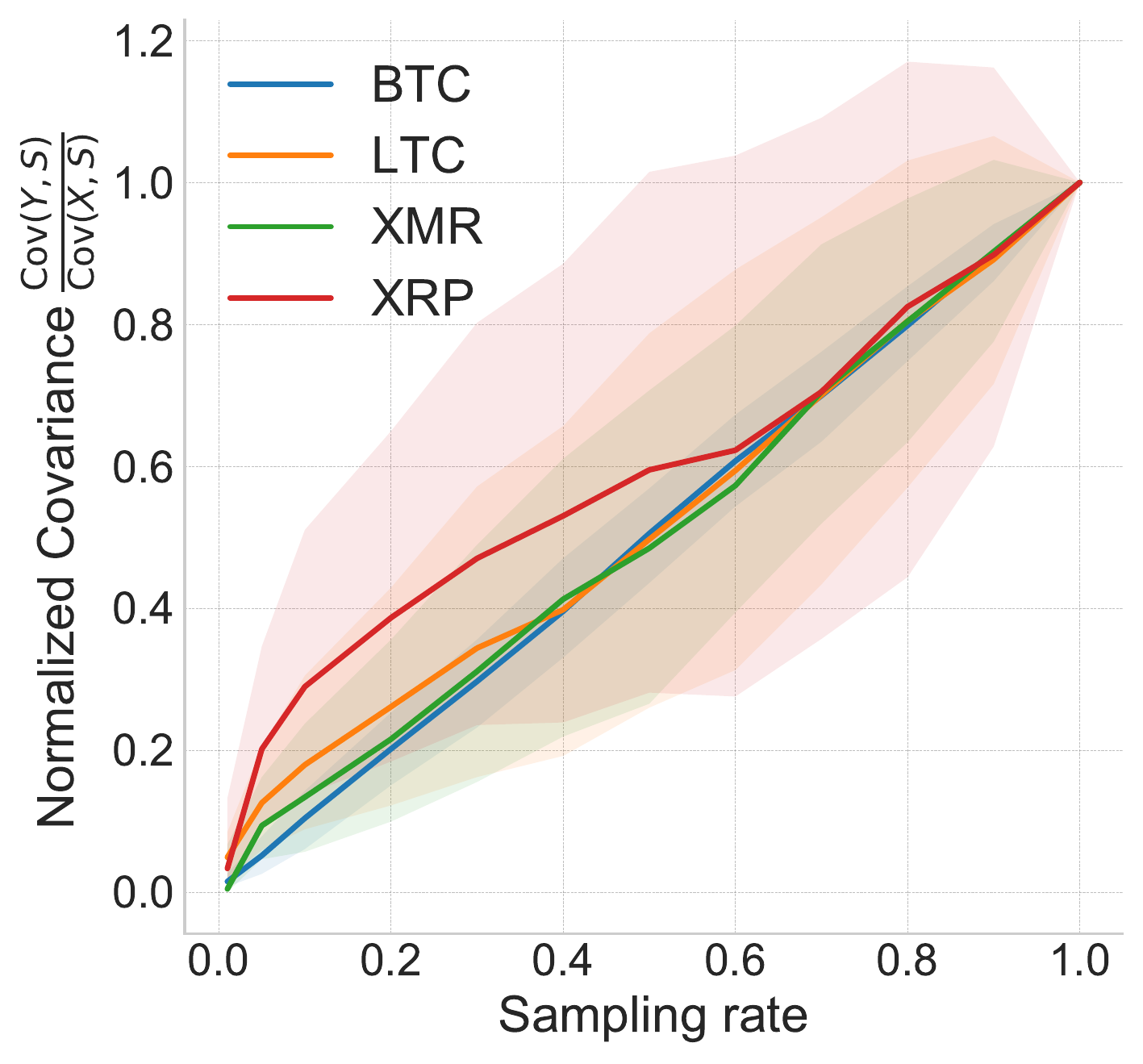}
%\hspace{0.15cm}
\includegraphics[align=c,width=0.49\columnwidth]{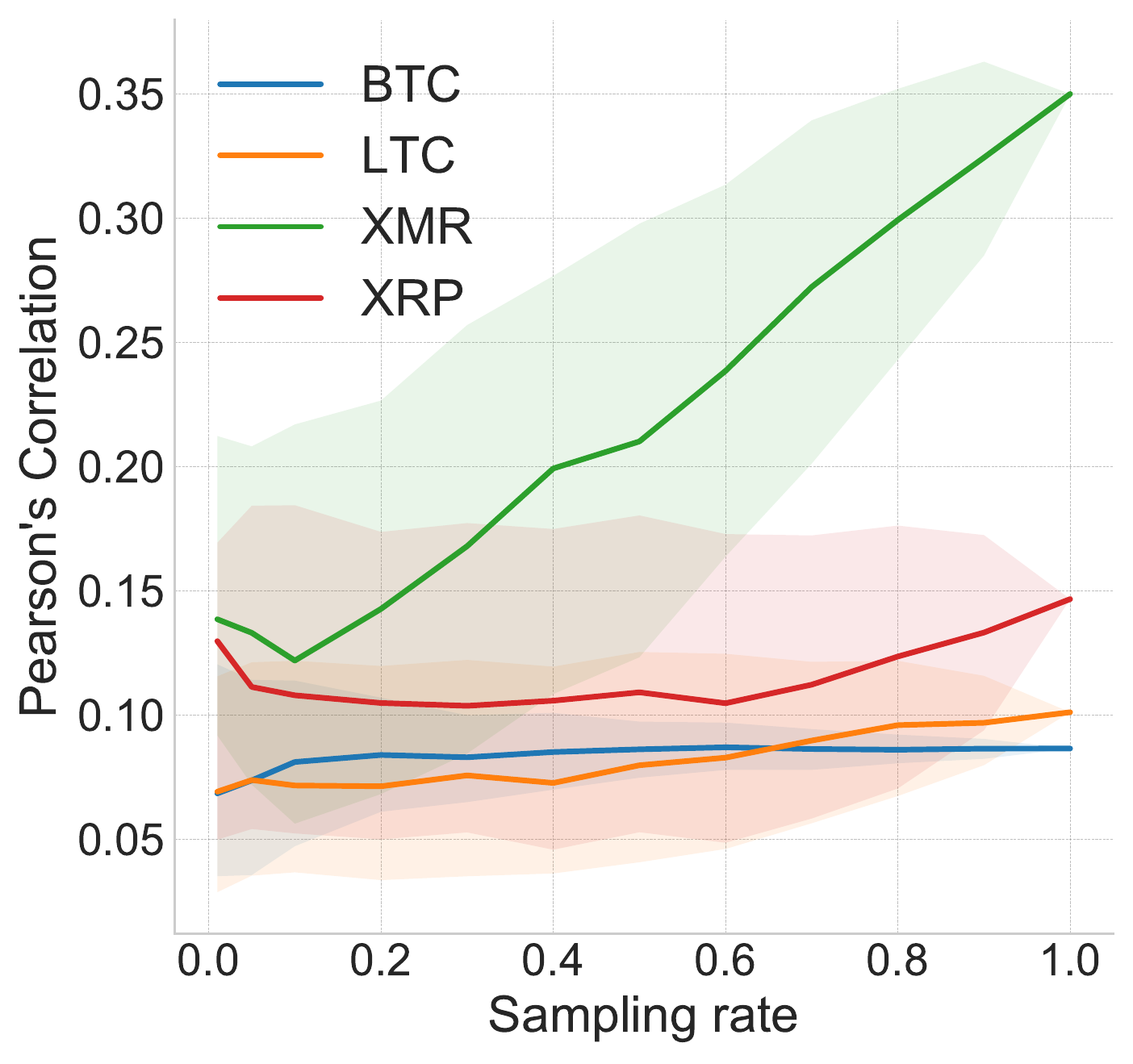}
\caption{Loss of correlation between cryptocurrencies repository popularity and their prices for different sampling rates. Each point is the median Pearson's correlation coefficient over $1000$ samples. Error bars show the standard deviation.
For each cryptocurrency, we calculated over the $1000$ samples, (\textbf{left}) the median  normalized covariance $\frac{\mathrm{Cov}(Y,S)}{\mathrm{Cov}(X,S)}$ and (\textbf{right}) the Pearson's correlation coefficient between the price and the popularity of related Github repositories at different sampling rates. Shaded regions mark the inter-quartile ranges for each coin.}
\label{fig:crypto_ext}
\end{figure}
%\end{center}
% We analyze the correlation between popularity of a technology
Figure~\ref{fig:crypto_ext} explores the effect that sampling has on the correlation. The left plot shows a clear decrease in the relative covariance for lower sampling rates, corroborating our theoretical results. As opposed to the behavior of influenza outbreaks (cf., Fig.~\ref{fig:exp_epidemics_ext_signal}), in Figure~\ref{fig:crypto_ext} (Right) we can see that a decay of covariance tends to induce a loss of correlation, especially for those coins with low variance relative to their mean. Supplementary Figure~\ref{fig:crypto_mi}  depicts a decrease in mutual information for BTC and LTC, while the other two coins are independent of the external signal.

% \note{Move text below to the end of the empirical section}
\paragraph{Synthetic Data.}
Finally we investigate the impact of sampling on the predictability of synthetic data generated by an auto-regressive process (\textit{SI}, \textit{Synthetic Data Experiments}). 
% \S\ref{appendix::Synthetic_Data}). 
In addition to autocorrelation and permutation entropy, we measure the error of forecasts made by an auto-regressive model trained on the sampled data. Similar to other metrics that demonstrate a loss of predictability, prediction error grows at lower sampling rates (\textit{SI},  Figure \ref{fig::RMSEvsDropout}). As a result, the forecasts made by auto-regressive models from  data collected at low sampling rates are no more accurate than forecasts made by a Poisson model that assumes independent events.  Sampling further distorts the observed dynamics of the auto-regressive process by introducing heteroskedasticity into the sampled time series. The time-varying variance causes predictions to deteriorate (\textit{SI}, \textit{Synthetic Data Experiments},
% \S\ref{appendix::Synthetic_Data}, 
Proposition~\ref{prop:garch}).

%%%%%%%%%%%%%%%%%%%%%%%%%%%%%%% METHODS %%%%%%%%%%%%%%%%%%%
\subsection*{Materials and Methods}
\paragraph{Permutation Entropy (PE).}
We use %normalized weighted 
permutation entropy %$H^\mathrm{P}_{(\textrm{w})}$ 
as a model-free measure of predictability of a time series% $\{x_t\}_{t=1}^N$
~\cite{Bandt2002permutation, weighted_Perm_Entropy,garland2014model}. Permutation entropy captures the complexity of a time series via statistics of its ordered sub-sequences of the type $s=[x_t,x_{t+\tau},\ldots,x_{t+(d-1)\tau} ]$, given embedding dimension $d$ and a temporal delay $\tau$. Let $\mathcal{S}_{d,\tau}$ be the collection of all $d!$ permutations $\pi$ of size $d$ and temporal delay $\tau$. For each $\pi\in\mathcal{S}_{d,\tau}$, we determine the relative frequency $P(\pi)$ of that permutation occurring in the time series.
The permutation entropy of order $d\ge 2$ and delay $\tau\ge1$ is defined as
\begin{equation}
\label{defn::permutation_entropy}
    H^{\mathrm{P}}(d,{\tau}) = -\sum_{\pi\in \mathcal{S}_{d,\tau}} P(\pi)\log_2P(\pi)
\end{equation}
We use weighted permutation entropy~\cite{weighted_Perm_Entropy} to lessen the influence of observational noise on the ordinal pattern of the signal, in which weights with respect to a sub-sequence with a certain ordinal pattern are introduced to reflect the importance of ordinal changes in large amplitudes.
Finally, we \textit{normalize} weighted permutation entropy by dividing it by $\log_2(d!)$, log of the number of possible permutations. See \textit{SI}, \textit{Permutation Entropy Criterion}, 
% \S\ref{sec:entropy} 
for a formal definition.
To estimate PE of a time series we need to specify the order $d$ and time delay $\tau$. The optimal parameters will depend on the specific properties of the time series, for example, the periodic behavior of the system relates to the delay parameter~\cite{riedl2013practical}. Here, we follow the approach described in \cite{scarpino2019predictability}, which performs a grid search over the pairs $(d,\tau)$, $2\leq d\leq 5$ and $1\leq \tau \leq 7$ searching for the values that minimize $H^{\mathrm{P}}(d,\tau)$. However, for the parameter search, PE is normalized by the number of observed permutations instead of the possible permutations, given that otherwise, $H^{\mathrm{P}}(d,\tau)$ is decreasing as a function of $d$. Finally, the parameters found for each ground truth signal are used to compute the PE of the corresponding sampled time series.

\paragraph{Mutual Information}
Mutual information characterizes the amount of information one random variable contains about another, specifically capturing the reduction in the uncertainty of one random variable due to the knowledge of the other. The mutual information between two random variables is defined as $\dis I(X;Y) = \mathbb{E}_{p(x,y)} \ln\frac{p(X,Y)}{p(X)p(Y)}$.

Here we consider the mutual information between two time series. We calculate the mutual information between two time series with \textit{PyInform} \cite{pyinform_asu}. %We observe that the mutual information between the external signal and the sampled signal has demonstrated a monotonically increasing tendency w.r.t.\ sampling rate. 

\paragraph{Loss of autocorrelation of the Sampled Signal.}
% Throughout the paper, without loss of generality, normal approximation will be used:
% $Y_t \sim B([X_t],p)$ is approximated by $Y_t \sim \mathcal{N}([X_t]p,[X_t]pq)$.
%Our first result shows that the higher the percentage of information that is omitted, the more the autocorrelation between event counts of early time and later time is erased.
Our first theoretical result shows that sampling reduces the auto-covariance of the observed signal, i.e., the covariance of the time series $Y$ and its time-lagged version.
\begin{theorem}
\label{prop::covariance_dropout_quadratic}
The time series $X$, $Y$ are related by $Y \sim B([X],p)$, where $B([X],p)$ is a Bernoulli random process with success rate $p$.

The covariance matrices $\boldsymbol{\Sigma}_X$ and $\boldsymbol{\Sigma}_Y$ are related as
\begin{equation}
\label{eqn::covariance_dropout_quadratic}
\boldsymbol{\Sigma}_Y = p^2\boldsymbol{\Sigma}_X + p(1-p)\E[X]\boldsymbol{I}
\end{equation}
where $\boldsymbol{I}$ is the identity matrix.
\end{theorem}

We can use the expression in Theorem~\ref{prop::covariance_dropout_quadratic} to approximate the autocorrelation of the sampled time series $Y$ as a function of the ground truth signal $X$. 
Autocorrelation is defined as Pearson correlation between values of the signal at different times, i.e.,
$\rho_{X_i,X_j} = \frac{\mathrm{Cov}(X_i,X_j)}{\sigma_{X_i}\sigma_{X_j}}.$
For sake of simplicity, %we approximate the variance of the process with a time-independent variance,
we assume that the ground truth process is stationary. 
\begin{corollary}
\label{coro::auto_corr}
The autocorrelation of  sampled time series $Y$ is 

\begin{equation}
\label{eq:autocorr}
    \rho_{Y_i,Y_j} \approx \frac{ p^2 \mathrm{Cov}(X_i,X_j)}{p^2\mathrm{Var}(X) + p(1-p)\E[X]}.
\end{equation}
\end{corollary}

\begin{corollary}
\label{coro::autocorrelation}
The magnitude of autocorrelation  $|\rho_{Y_i,Y_j}|$ of the observed signal $Y$, increases monotonically as a function of the sampling rate $p$.
\end{corollary}

\begin{corollary}
\label{coro::covariance}
 The covariance between the observed signal $Y$ and an arbitrary external signal $S$
is related to the covariance between the ground truth signal $X$ and the same external signal $S$ by,
\begin{equation}\label{eq:cov_ex}
\mathrm{Cov}(Y,S) =  p\,\mathrm{Cov}(X,S).
\end{equation}
\end{corollary}

\paragraph{Epidemics Data.}
Weekly state-level data for all diseases was obtained from Scarpino \& Petri \cite{scarpino2019predictability} and originally compiled by the USA National Notifiable Diseases Surveillance System (see \textit{SM}, Table~\ref{si-epidemics} for statistics of the data).  For the covariance experiment, we used influenza data from 2010-2015 obtained for the US Outpatient Influenza-like Illness Surveillance Network (ILINet) that overlaps with Google Flu Trends Data.

\paragraph{Twitter Data.}
%\ef{All the numbers here seem inconsistent }
The social media data used in this study was collected from Twitter in 2014. Starting with a set of 100 users who were active discussing ballot initiatives during the 2012 California election, we expanded this set by retrieving the accounts of the users they followed, for a total of 5,599 \textit{seed users}. We collected all posts made by the seed users and their friends (i.e., users they followed on Twitter) over the period of June--November 2014, a total of over 600 thousand users. % and 18 million  hashtags. %This represents the complete activity of the 600K Twitter user population under study.
We extracted time series of the activity for 100 most popular hashtags and 150 most active users in this data (see \textit{SM}, Tables~\ref{si-twitter-S1} and \ref{si-twitter-S2} for statistics of the data).  

\paragraph{GitHub Data.}
The GitHub data we analyzed contains anonymized records of user activities over a time period spanning from January 1st, 2015 to March 31st, 2015.
The activities represent the  actions users performed on the repositories, including  watching the  repositories to receive notifications about project activity. We used \textit{watches}, \textit{forks}, and \textit{create} event activity as a measure of popularity of a repository in Github. Overall, our dataset captures $43,962$ Github activity events by $5,509$ users on $2,036$ repositories (see \textit{Supplementary Information (SI)}, Table~\ref{si-github} for additional statistics). Cryptocurrencies' historical prices were obtained from publicly available Kaggle datasets.

%%%%%%%%%%%%%%%%%%%%%%%%%%%%%%% DISCUSSION %%%%%%%%%%%%%%%%%%%%%%%%%%%%

\section*{Discussion}
We presented a framework to analyze the effects of partial observation of a dynamic process, showing that sampling degrades the predictability of the process. Using empirical data from three domains, namely epidemics, social systems, and software collaborations, we highlighted how this fundamental predictability limit affects the forecasting of disease outbreaks, social media content popularity, and emergence of cryptocurrency technologies.
%We examined the impact of information loss due to the sampling process of a time series.
%In addition to a decrease autocorrelation, we show that the signal also changes qualitatively, acquiring different features from the unfiltered signal, and not even informative external signals can help recover that information loss.
We showed that even when events making up the temporal signal are sampled at random, sampling qualitatively changes the observed dynamics of the process, decreasing the autocorrelation and increasing permutation entropy. Moreover, the predictability loss is irreversible: even a highly informative external signal does not help to fully recover predictability lost to sampling. These findings were corroborated by experiments on synthetic data.

Our work is motivated by applications requiring the forecasting of partially observed, or sampled, complex systems. Such situations may occur, for example, when country-wide forecasts of influenza have to be made based on reports by a few hospitals; when longitudinal opinion polls of a population are used to predict an election; or when researchers avail of random samples of social media activity to characterize complex social dynamics.
Beyond prediction, models learned from data can also elucidate social behaviors~\cite{Hofman2017}. Scientists developed techniques for temporal data analysis, based on anomaly detection~\cite{dewhurst2019shocklet} and regression discontinuity design~\cite{Herlands2018kdd}, to uncover natural experiments that yield insights into the mechanisms of human decision making. As we showed in this paper, however, these techniques may be systematically biased by temporal sampling. It is, therefore, imperative to account for potential 
%We emphasize the importance of accounting for  possible
sampling biases in the study of social dynamics, so that no results are erroneously attributed to the phenomena under study. Thus, it is important for future research to focus on statistical tools and sampling methods that can correct for these possible biases. 

%\note{AA new paragraph below}{}
% \paragraph{Causal Inference for Intervention Studies}
Our work suggests that partial observability not only diminishes the predictability of a dynamic process, but also introduces a source of heterogeneous random noise that can potentially %remove any temporal confounders from the data.  Underestimating temporal confounders threatens the validity of causal inferences and can lead to erroneous conclusions. 
mislead causal inference methods and threaten their validity.
For example, interrupted time series (ITS) analyses is one of the most widely applied approaches to evaluate natural experiments in health interventions~\cite{craig2017natural}. ITS consists of a sequence of counts over time, with one or more well-defined change points that correspond to the introduction of an intervention. The effect of the intervention can be estimated by fitting a linear regression model with a dummy variable for the before/after intervention, and additional variables to control for time-varying confounders. 
Only recently, researches have addressed methodological issues associated with ITS analysis caused by over-dispersion of time series data and autocorrelation~\cite{bernal2017interrupted}. For instance, a study estimating the impact of a ban on the offer of multipurchase discounts by retailers in Scotland, found a 2\% decrease in alcohol sales after controlling for seasonal autocorrelation, compared with a previous study's finding no impact \cite{robinson2014evaluating}.
Our work provides a theoretical framework to understand and quantify new sources of biases that sampling creates that can affect intervention studies.
% It is therefore essential to consider the biases introduced by partial observability in ITS and similar analyses, which include the potential for erroneous conclusion of intervention studies. 

%%%%%%%%%%%%%%%%%%%%%%%%%%%%%%%%%%%%%%%%%%%%
\subsection*{Code Availability}\small{
Codes to generate the results of the paper are available on \url{https://github.com/aabeliuk/Predictability-partially-observed}. 
}

\subsection*{Data Availability} \small{
This work uses publicly available data. Links to data repositories can be found in the Methods section.
}

\subsection*{Acknowledgments}\small{
The authors thank Linhong Zhu for collecting the Twitter data and America Mazuela for the illustration.
% Authors thank Jeremy Abramson for creating the malicious email data and Linhong Zhu for collecting the Twitter data.
% 
\textbf{Funding:}
This work was supported by the Office of the \textit{Director of National Intelligence} (ODNI) and the \textit{Intelligence Advanced Research Projects Activity} (IARPA) via the \textit{Air Force Research Laboratory} (AFRL) contract number FA8750-16-C- 0112, and by the \textit{Defense Advanced Research Projects Agency} (DARPA), contract number W911NF-17-C-0094. The U.S. Government is authorized to reproduce and distribute reprints for Governmental purposes notwithstanding any copyright annotation thereon. Disclaimer: The views and conclusions contained herein are those of the authors and should not be interpreted as necessarily representing the official policies or endorsements, either expressed or implied, of ODNI, IARPA, AFRL, DARPA, or the U.S. Government.
\textbf{Authors contributions:}
All authors conceptualized the study; ZH carried out formal analysis; AA carried out validations with empirical data; ZH and AA carried out validations with synthetic data; all authors contributed to writing and reviewing the manuscript.
\textbf{Competing interests:}
Authors declare no competing interests.
%\textbf{Data and materials availability}
%All data is available.
}

% \newpage
\bibliography{scibib}

\bibliographystyle{Science}

% \section*{List of Supplementary Materials}

% \noindent Materials and Methods

% \noindent Supplementary Text

% \noindent Figs. 6 to 15

% \noindent Tables 1 to 4

%%%%%%%%%%%%%%%%%%%%%%%%%%%%%%%%%%%%%%%%%%%%%%%%%%%%%%%%%%%%%%%%%%%%%%%%%%% APPENDIX
\newpage
%\section*{Supplementary Information (SI) Appendix}
\section*{Supplementary Information (SI)}

%%% PROOFS
\subsection*{Proofs to Propositions and Corollaries}
\label{sec::appendix}
\subsubsection*{Proof of Theorem \ref{prop::covariance_dropout_quadratic}}
\label{sec:proof1}
% \note{KL: Add theorem here.} 

\begin{theorem}[Restatement of Theorem \ref{prop::covariance_dropout_quadratic}]

The time series $X$, $Y$ are related by $Y \sim B([X],p)$, where $B([X],p)$ is a Bernoulli random process with success rate $p$. The covariance matrices $\boldsymbol{\Sigma}_X$ and $\boldsymbol{\Sigma}_Y$ are related as
\begin{equation*}
\boldsymbol{\Sigma}_Y = p^2\boldsymbol{\Sigma}_X + p(1-p)\E[X]\boldsymbol{I}
\end{equation*}
where $\boldsymbol{I}$ is the identity matrix.

\end{theorem}

\begin{proof}
First, we compute the off-diagonal elements of covariance matrices $\boldsymbol{\Sigma}_Y$ and $\boldsymbol{\Sigma}_X$, i.e., the relation between $\mathrm{Cov}(X_i,X_j)$ and $\mathrm{Cov}(Y_i,Y_j)$.
\begin{align}
\mathrm{Cov}(Y_i,Y_j) &= \E[Y_i Y_j] - \E[Y_i]\E[Y_j] \nonumber \\
&= \E_X\big[\E_Y[Y_i Y_j|X_i,X_j]\big] - \E_X\big[\E_Y[Y_i|X_i]\big]\E_X\big[\E_Y[Y_j|X_j]\big] \nonumber\\
&= \E_X\big[\E_Y[Y_i|X_i]\E_Y[Y_j|X_j]\big] - \nonumber\\
&\phantom{{}=1}\E_X\big[\E_Y[Y_i|X_i]\big]\E_X\big[\E_Y[Y_j|X_j]\big] \nonumber\\
&= \E_X\big[p^2[X_i][X_j]\big] - \E_X\big[p[X_i]\big]\E_X\big[p[X_j]\big] \nonumber \\
&= p^2\big( \E\big[[X_i][X_j]\big] - \E\big[[X_i]\big]\E\big[[X_j]\big]	\big) \nonumber\\
&\approx p^2 \mathrm{Cov}(X_i,X_j)
\label{cov_1}
\end{align}
Next, we discuss diagonal elements, i.e., the relation between $\mathrm{Var}[Y]$ and $\mathrm{Var}[X]$. 
Without loss of generality, normal approximation will be used:
$Y \sim B([X],p)$ is approximated by $Y \sim \mathcal{N}([X]p,[X]p(1-p))$.
Thus, \textit{for fixed} $X$, $\displaystyle Z \equiv \frac{Y - pX}{\sqrt{p(1-p)X}}\sim\mathcal{N}(0,1)$, then $\displaystyle Z^2 = \frac{Y^2 - 2pYX+p^2X^2}{p(1-p)X}\sim\chi^2(1)$, i.e., \[ \E[Z^2] = \E\bigg[\frac{Y^2 - 2pYX+p^2X^2}{p(1-p)X}\bigg]=1\]
This gives $\E [Y^2|X] = p(1-p)X + 2p\E[Y|X]X - p^2X^2 = p(1-p)X + p^2X^2$
\begin{align}
\mathrm{Var}[Y] &= \mathrm{Cov}(Y,Y) = \E[Y^2] - \E[Y]^2 \nonumber\\
&= \E_X\big[\E_Y[Y^2|X]\big] - \E_X\big[\E_Y[Y|X]\big]^2\nonumber\\
&= \E_X[p(1-p)X + p^2X^2] - \E_X[pX]^2\nonumber\\
&= \E_X[p(1-p)X + p^2X^2] - p^2\E_X[X]^2\nonumber\\
&= p(1-p)\E[X] + p^2\mathrm{Var}[X]
% &= \boxed{p^2\mathrm{Cov}(X,X) + p(1-p)\E[X]}
\label{cov_2}
\end{align}
Equations \eqref{cov_1} and \eqref{cov_2} give the desired result.
\end{proof}

\subsubsection*{Proof of Corollary \ref{coro::auto_corr}}
\begin{proof}
Autocorrelation is defined as Pearson correlation between values of the signal at different times, i.e.,
$$\rho_{X_i,X_j} = \frac{\mathrm{Cov}(X_i,X_j)}{\sigma_{X_i}\sigma_{X_j}}.$$
 This yields the following expression for autocorrelation of the time series $Y$:
\begin{equation}
    \rho_{Y_i,Y_j} = \frac{\mathrm{Cov}(Y_i,Y_j)}{\sigma_{Y_i}\sigma_{Y_j}}
    =\frac{ p^2 \mathrm{Cov}(X_i,X_j)}{\sqrt{p^2\mathrm{Var}(X_i) + p(1-p)\E[X_i]}\sqrt{p^2\mathrm{Var}(X_j) + p(1-p)\E[X_j]}}.
\end{equation}
The last equality comes from replacing Equation~\ref{cov_1} in the numerator and Equation~\ref{cov_2} in the denominator.
Finally, we assume the ground truth process is stationary, i.e., the process has a time-independent variance ($\mathrm{Var}(X_j)\approx \mathrm{Var}(X_j)\; \forall i,j$) and mean ($\E(X_j)\approx \E(X_j)\; \forall i,j$), yielding the desired result:
\begin{equation}
 \rho_{Y_i,Y_j} \approx \frac{ p^2 \mathrm{Cov}(X_i,X_j)}{p^2\mathrm{Var}(X) + p(1-p)\E[X]}.
\end{equation}
\end{proof}

\subsubsection*{Proof of Corollary
\ref{coro::autocorrelation}}
\begin{proof}
Based on Corollary~\ref{coro::auto_corr}, the autocorrelation of the sampled signal $Y$ between times $i$ and $j$, which is a function of sampling rate $p$, is defined as
\begin{equation}
R(p)=\rho_{Y_i,Y_j}
=\frac{ p^2 \mathrm{Cov}(X_i,X_j)}{\sqrt{p^2\mathrm{Var}(X_i) + p(1-p)\E[X_i]}\sqrt{p^2\mathrm{Var}(X_j) + p(1-p)\E[X_j]}}.
\end{equation}
The autocorrelation lies in the range $[-1, 1]$, and we want to prove that its magnitude increases as a function of $p$. Hence,  next we show that $\frac{\mathrm{d}}{\mathrm{d}p}R^2(p)\geq 0, \forall 0\leq p\leq 1$.

\begin{equation}
\frac{\mathrm{d}}{\mathrm{d}p}R^2(p) =
    p\;\mathrm{Cov}(X_i,X_j)^2 \frac{ \mathrm{Var}(X_i) \E[X_j] p + \E[X_i] \mathrm{Var}(X_j) p + 2 \E[X_i] \E[X_j] (1 - p)}{\left(\mathrm{Var}(X_i) p  - \E[X_i] p  + \E[X_i]\right)^2 \left(\mathrm{Var}(X_j) p - \E[X_j] p + \E[X_j]\right)^2},
\end{equation}
where both the numerator and denominator are trivially positive for all values in $0\leq p\leq 1$ given that $\E[X]\geq 0$, and the result follows.
\end{proof}

\subsubsection*{Proof of Corollary \ref{coro::covariance}}
\begin{proof}
We compute the covariance between sampled signal $Y$ and external signal $S$ as
\begin{align}
\mathrm{Cov}(Y,S) &= \E (Y-\E Y)(S - \E S) \nonumber\\
&= \E [Y S] - \E[Y] \E[S]\nonumber\\
&= \E_{X,S}\big[S\E_Y[Y|S,X]\big] - \E_X\big[\E_Y[Y|X]\big] \E[S]\nonumber\\
&= \E_{X,S}[S Xp] - \E_X[Xp]\E[S]\nonumber\\
&= p\big(\E[S X] - \E[X]\E[S]\big)\nonumber\\
&= p\,\mathrm{Cov}(X,S)
\end{align}
\end{proof}

\setcounter{proposition}{1}
\renewcommand{\theproposition}{S\arabic{proposition}}

\subsection*{Permutation Entropy Criterion}
%\paragraph*{Entropy Characterization for Information Loss}
% Increase in Permutation Entropy}
\label{sec:entropy}
% As another criterion to quantify predictability of the time series, we also compute the normalized weighted permutation entropy $H^\mathrm{P}_{(\textrm{w})}$ \cite{weighted_Perm_Entropy,garland2014model} of the synthetic signal in the information loss setting.
Entropy measures the uncertainty of a random variable, which intuitively serves as an indicator of predictability of a  stochastic event. In statistical physics, entropy characterizes the amount of possible microscopic state in a system, thus the more microscopic states exist in a system, the more chaotic a system is, and the harder it becomes to predict its behavior. 

\begin{definition}[Shannon entropy]
For a random variable, the (Shannon) entropy is defined as
\begin{align}
    \label{eqn::entropy_defn}
    \textrm{discrete case: }H(p) &= -\sum_{x\in\mathcal{X}} p(x)\ln p(x) \\
    \textrm{continuous case: }H(p) &= -\int_S p(x)\ln p(x)\,\mathrm{d} x
\end{align}
$p(x)$ is the probability distribution of the random variable. In the discrete case, $\mathcal{X}$ is the collection of all possible values of the random variable. In the continuous case, $S$ is the support set of the random variable.
\end{definition}

In both discrete and continuous case, the larger the entropy value is, the more uncertain a random variable is, thus rendering the stochastic event a random variable represents harder to predict. A variant form of entropy is the \textit{permutation entropy}, which depicts the complexity of a time series through the statistics of %the orderings of 
of the values of its subsequences using ordinal analysis. The complexity can be interpreted as the diversity of the trends among the subsequences of certain length. Therefore, the higher the entropy, the more different trends exist in the time series, which renders its prediction more difficult.

Permutation entropy has been used in various fields to characterize the predictability of time series under interest \cite{garland2014model,politi2017quantifying}. Interestingly, this quantity has also been used as a forensic tool to inspect and identify potential corruption in the source data \cite{garland2018anomaly}.  
% Using sliding window calculation to evaluate permutation entropy of isotope record time series, Garland et al.\ observed the abrupt increase in the time frame under consideration, thus suspecting that the instrument used in field work of collecting the raw data was compromised, which they confirmed as their conclusion.

\begin{definition}[Permutation Entropy]%\cite{garland2014model}
Given a time series $\{x_t\}_{t=1}^N$. Let $\mathcal{S}_d$ be the collection of all $d!$ permutations $\pi$ of order $d$. For each $\pi\in\mathcal{S}_d$, determine the relative frequency of that permutation occurring in $\{x_t\}_{t=1}^N$:
\[
P(\pi) = \frac{\mathrm{Card}[\{t \,|\, t\le N-d,\phi(x_{t+1},\cdots,x_{t+d})=\pi\}]}{N-d+1} = \sum_{t\le N-d} \frac{1}{N-d+1} \delta\big( \phi(x_{t+1}^{(d)}),\pi \big)
\]
where $P(\pi)$ quantifies the frequency of an ordinal pattern $\pi$, and $\delta(a,b) = \begin{cases}
1, \, & \textrm{if } a=b\\
0, \, & \textrm{if } a\ne b
\end{cases}$.

The permutation entropy of order $d\ge 2$ is defined as
\begin{equation}
    H^{\mathrm{P}}(d) = -\sum_{\pi\in \mathcal{S}_d} P(\pi)\log_2P(\pi)
\end{equation}
\end{definition}

The ordinal pattern means the relative magnitude relation among successive time series values. As an example, if $x_1= 3, x_2 = 6, x_3 = 1$, then the ordinal pattern of this subsequence $\{x_1,x_2,x_3\}$ is $\phi(x_1,x_2,x_3) = (312)$ because $x_3\le x_1\le x_2$.

Besides the order $d$, the more general definition of permutation entropy  \eqref{defn::permutation_entropy} has one more parameter: temporal delay $\tau$. The ordinal pattern can be defined in the same way with respect to the subsequence $x_t, x_{t+1\tau}, x_{t+2\tau},\cdots, x_{t+(d-1)\tau}$, which gives permutation entropy $H^\mathrm{P}(d,\tau)$. To facilitate interpretation, we present results from continuous intervals by fixing $\tau = 1$.

To lessen the influence of observational noise on the ordinal pattern of the signal, the \textit{weight} w.r.t.\ a subsequence with certain ordinal pattern is introduced to reflect the importance of ordinal changes in large amplitude. For a subsequence of length/order $d$ consisting times series values from $x_{t+1}$ to $x_{t+d}$, which is denoted as $x_{t+1}^{(d)}$ with arithmetic mean value $\bar{x}_{t+1}^{(d)}$, its \textit{weight} is defined \cite{weighted_Perm_Entropy} as
\[
w(x_{t+1}^{(d)}) = \frac{1}{d} \sum_{j=t+1}^{t+d} (x_j - \bar{x}_{t+1}^{(d)})^2
\]
As a result, the \textit{weighted frequency of a permutation} is defined as
\[
P_w(\pi) =\sum_{t\le N-d} \bigg( \frac{w(x_{t+1}^{(d)})}{\sum_{t'\le N-d}w(x_{t'+1}^{(d)})} \bigg) \delta\big( \phi(x_{t+1}^{(d)}),\pi \big)
\]

The \textit{weighted permutation entropy} is defined as
\begin{equation}
    H^{\mathrm{P}}_{(w)}(d) = -\sum_{\pi\in\mathcal{S}_d} P_w(\pi)\log_2 P_w(\pi)
\end{equation}
Here we \textit{normalize} the weighted permutation entropy by the log-number of  the factorial of the observed permutations. Thus, the weighted permutation entropy takes value between 0 and 1.

%%%%%%%%%%%% SIMULATION

\subsection*{Synthetic Data Experiments}
\label{appendix::Synthetic_Data}
Here, we validate our findings on synthetically generated time series data. We show that predictability diminishes as data is lost to sampling.
We first consider an idealized scenario, where $X$ represents an autoregressive process, from which events are sampled at random to create the observed sampled signal $Y$.
% Details of generating synthetic dataset are presented in appendix section \ref{appendix::materials_methods}.

\subsubsection*{Synthetic Time Series Generation}
\label{sec-si-data}
\paragraph{External Signal}
First, we generate an external signal $S$. To assure autocorrelation, we generate it using the autoregressive integrated moving average (ARIMA)  model:
\begin{equation}
    S + \sum_{i=1}^k \alpha_i^{\textrm{ES}} S_{t-i} = \varepsilon_t + \sum_{j=1}^l \beta_j^{\textrm{ES}} \varepsilon_{t-j} .
\end{equation}
ARIMA coefficients satisfy the stationarity conditions, so that the external signal $S$ is second-order stationary. The stationarity conditions require that all roots of the polynomials $\alpha(x) = 1+\alpha_1 x + \cdots + \alpha_k x^k$ and $\beta(x) = 1+\beta_1 x + \cdots + b_l x^l$ satisfy $|z|>1$; i.e., all roots of these two polynomials are located outside the unit disk.
We enforce the second order stationarity by determining roots of $\alpha(x)$ and $\beta(x)$ first and then solving for the corresponding regression coefficients $\{\alpha_i\}$ and $\{\beta_j\}$, which specify the model.

\paragraph{Ground Truth Signal}
We generate the ground truth signal $X$ in a similar manner, except that the generation model entails a term for the external signal $S$. This ensures that the ground truth and the external signals are correlated. Specifically, we assume the ground truth signal is defined as:
\begin{equation}
X + \sum_{i=1}^K\alpha_i^{\textrm{GT}} X_{t_i} = \varepsilon_t + \sum_{j=1}^L \beta_j^{\textrm{GT}} \varepsilon_{t-j} + S.
\end{equation}
\noindent
%In the case of no external signal, we can set $S = 0$ for all $t$.
When generating the ground truth signal time series, we require that the \textit{autocorrelation} of the ground truth signal is strong enough so that it can be distinguished from the random white noise.

%The motivation of generating ground truth signal in this way is to introduce autocorrelation between external signal and ground truth signal.

\paragraph{Observed signal} We sample the ground truth data $X$ to obtain the time series of observed events, $Y$. The sampling rate $p$ characterizes the probability of sampling an event. Because the count process is described by the ARIMA model, which inevitably gives real-number-valued count instead of integer-valued count, the decimal part of the count $X$ is treated as a separate instance, and if the decision is made to keep it, its original value will be added to the posterior data.
In other words, each instance of $Y$ obeys the %Bernoulli distribution
Binomial distribution $B(X,p)$.

Due to the stochastic nature of sampling, we generate ten different samples $Y$ based on the ground truth signal with the same sampling rate $p$. In each experiment, the sampled signal $Y$ is split into a training and testing data set, with training data used to train a predictor $\hat{Y}$ to predict the test data. The  accuracy of the predictor for a given sampling rate is then averaged over the ten experiments.
% \note{KL: Check $X \xrightarrow Y$}

%\paragraph{Grid Search and Stepwise Predictor}
\paragraph{Model Training}
Training an ARIMA model consists of two steps: First, a grid search is performed to find the best  hyper-parameters $(k,l)$, where $k$ is the order of the autoregressive model, and $l$ is the order of the moving-average model.
For each input signal, we search over the grid for the set of parameters resulting in the lowest AIC score. Next, the corresponding coefficients of the fixed-order ARIMA model are fitted to the data.

After the best ARIMA order parameters are determined, we do step by step prediction over a specified time range. At each prediction step, the data from the previous step is incorporated into the known data as new input signal, and consequently, the model is retrained to find the updated parameters. %coefficients of the fixed-order ARMA model are updated at every prediction step.

\paragraph{Prediction}
 In order to compare predictions at different sampling rates, we use normalized rooted-mean-square error (NRMSE) to measure how accurately we predict the observable $Y$.
Given the predicted values $\hat{y}_t$ with respect to time series $Y$, NRMSE is defined as
\begin{align}
    \mathrm{NRMSE}(y_t,\hat{y}_t) &= \frac{1}{\bar{y}}\sqrt{\frac{\sum_{i=1}^T(y_i-\hat{y}_i)^2}{T}}\nonumber \\
    &= \frac{1}{\frac{1}{T}\sum_{i=1}^T y_i}\sqrt{\frac{\sum_{i=1}^T(y_i-\hat{y}_i)^2}{T}}
\end{align}

\begin{center}
\begin{figure}[t]
\centering
\includegraphics[width=0.8 \columnwidth]{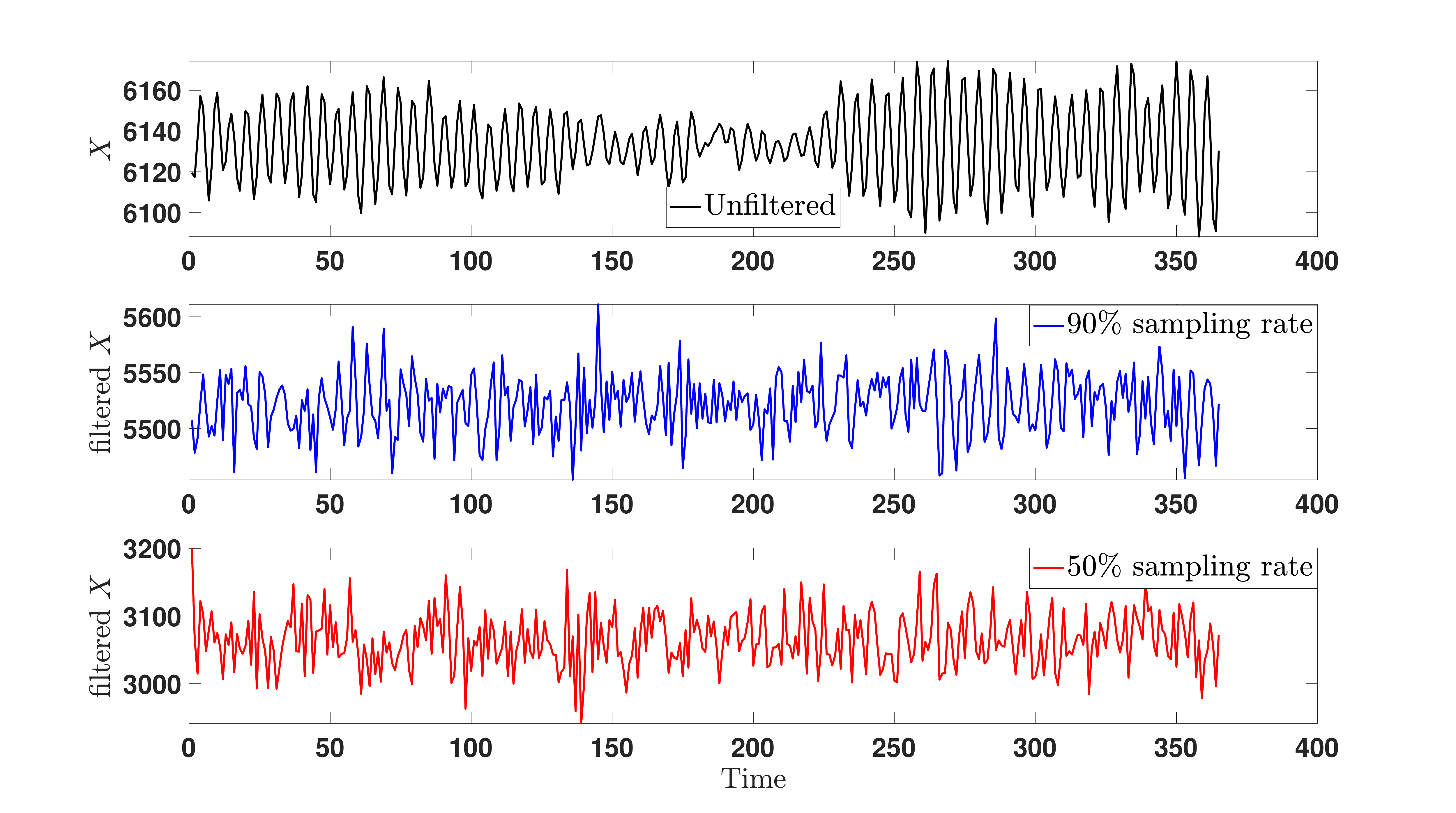}
\caption{Synthetic time series showing ground truth signal (top) and the observed  signal sampled at sampling rates $90\%$ (middle) and $50\%$ (bottom).
%\note{Change the lower graph here to sampled signals} \textcolor{blue}{addressed}
}
\label{fig:signals}
\end{figure}
\end{center}

%\subsubsection{Experiment Instance of Synthetic Data}
\subsubsection*{Numeric Experiments}
First, we illustrate all our theoretical claims using one instance of an ARIMA process generated with an external signal. Second, we present aggregated results of the prediction task, using multiple randomly generated ARIMA time series.

For the first set of results on the synthetic data, we generated an external signal $S$ with the ARIMA of order (3,0,2) and length 365, representing a full year of event counts.

Meanwhile, the ground truth signal $X$ assumes the ARIMA order(5,0,1).

\noindent
Figure~\ref{fig:signals} shows the ground truth and the sampled signals. Notice that with a simple visual inspection of the plot, one can observe that many of the temporal patterns present in the unsampled data seem to have disappeared in the sampled signal.
%\paragraph{Cushion Setting} We want to predict the count $Y$ at a future time; i.e., predict $Y$ at time $t-\tilde{c}$. Therefore, there is a cushion period $\tilde{c}$ between available ground truth data and the desired prediction range.

\paragraph{Predictors}
We train three predictors for $Y$, each of which can be used with or without an external signal. The predictors are:

\begin{description}
    % \item[Poisson predictor 1] \note{This predictor assumes that events in $Y$ are generated independently of each other at some rate.} The predictor estimates the intensity of the Poisson process as the average of the past 30 counts of the signal.
    \item[Poisson predictor] {assumes that events in $Y$ are generated independently of each other at some rate.} This predictor estimates the Poisson intensity as the average of counts of all available past data.
    \item[$\boldsymbol{\hat{Y}_t}$] uses the sampled signal to predict the observable $Y$.
    \item[$\boldsymbol{\hat{X}_t}$] uses the unsampled signal to predict the observable $Y$. i.e., the ARIMA parameters are fitted to $X$, and then used to predict the future values of $Y$ given its past values.
\end{description}

\begin{figure}[t]
\centering
\includegraphics[width=0.7 \columnwidth]{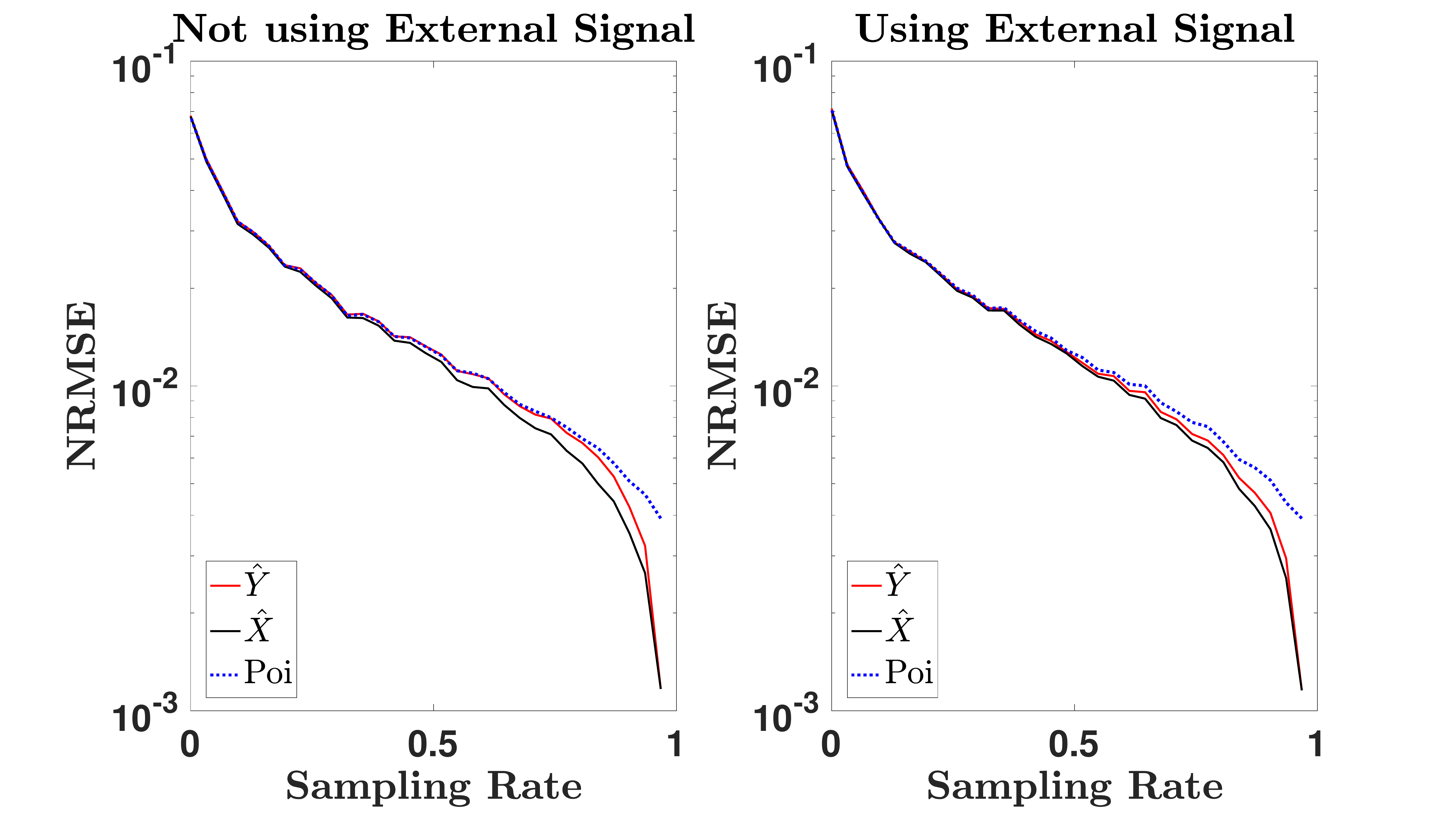}
\caption{Prediction accuracy in terms of NRMSE, and the difference in prediction performance among different predictors versus Sampling Rate Plot. Here we use three predictors to estimate sampled counting process $Y$. %The scale of both $x$ and $y$ axis is in log scale.
%\note{KL: Please make the same y-axis on both plots. The point here is that Poisson predictor has the same NMRSE with or without external signal, while ARIMAX improves. Also, it would look better to have two square plots next to each other horizontally, like Fig 6. } \textcolor{blue}{addressed}
%\note{Is there difference between the two left plots? Increase axis label fonts, put y-axis on the log scale --> no need to show the right hand plots. Change x-axis text to 'dropout rate'} \textcolor{blue}{addressed}
}
\label{fig::RMSEvsDropout}
\end{figure}

%\paragraph{Analysis on Figure \ref{fig::RMSEvsDropout}}
\paragraph{Prediction Accuracy}
Figure \ref{fig::RMSEvsDropout} shows normalized prediction errors (normalized RMSE) as a function of sampling rate to demonstrate the nonlinear decrease of the prediction error. As sampling rate decreases, prediction error grows.
We study the performance of the predictor $\hat{Y}$, which is trained on the history of the observed signal $Y$, as it is often employed in practice.

Performance of the predictor $\hat{X}$, trained on the full signal $X$, is almost always better (lower NRMSE) than performance of predictor $\hat{Y}$ (the plot show difference between predictors on the log scale). This phenomenon reveals that sampling weakens prediction accuracy. Moreover, our results suggest that the sampled process' increased noise and low autocorrelation obfuscates the underlying dynamic, making it harder to be described by an ARIMA model. %\note{AA:any other ideas of why $\hat{X}$ performs better? }
Using an informative external signal $S$ in prediction helps recover some of the lost information, shrinking the gap between $\hat{X}$ and $\hat{Y}$ predictors, as well as the overall prediction error.
%Meanwhile, the advantage of knowing the full ground truth signal in prediction diminishes significantly when the practitioner has the knowledge of the external signal $S$.

When little information is lost (i.e., at high sampling rate), predictors $\hat{X}$ and $\hat{Y}$ outperform the Poisson predictor, since they are able to leverage the autocorrelation of the signal with the ARIMA model. In addition, by comparing the gaps between predictors at high sampling rates (note the log scale), we see that adding external signal makes the Poisson predictor less competitive than the other predictors.
On the other hand,  Poisson predictor performs almost as well as $\hat{Y}$ and $\hat{X}$ when much of the information is lost (i.e., at low sampling rate).  This indicates that Poisson predictors are strong baselines for the observable $Y$ at low sampling rates. %The better accuracy demonstrated by the ARIMA model over baseline Poisson predictors are mostly restricted to the right end of sampling rate spectrum.

\begin{figure}[t]
\centering
\includegraphics[width=0.7\columnwidth]{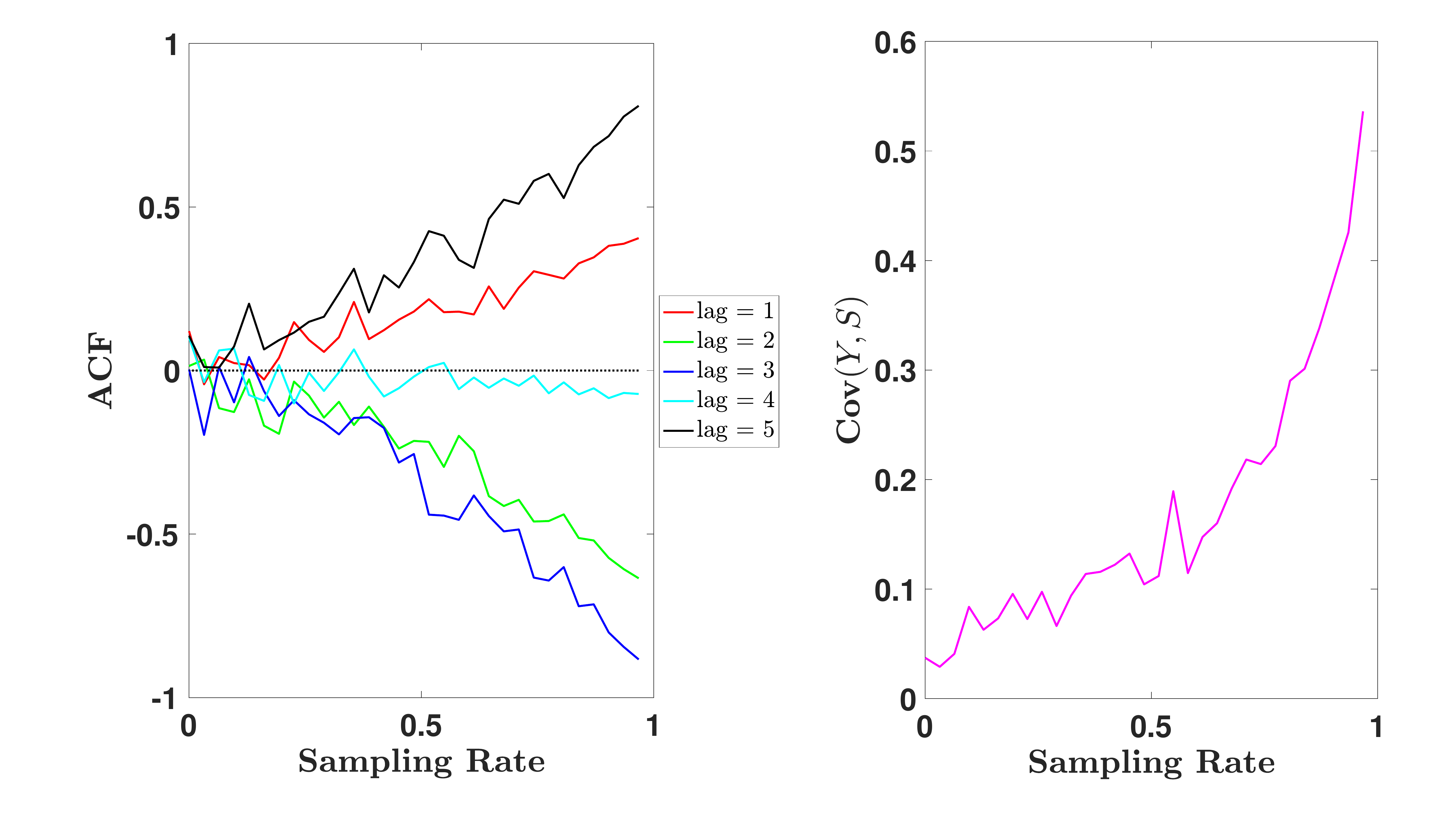}
\caption{Loss of predictability due to sampling. (left) autocorrelation of sampled time series decreases at low sampling rates. %The left plot shows the case without using external signal for prediction, while the lower 
(right) Covariance of the  external signal and observed signal decreases at low sampling rates.
%\note{AA: The second plot should compare the correlation between the external signal and ground truth.} \textcolor{blue}{addressed}
}
\label{fig::PACF}
\end{figure}

\paragraph{Loss of autocorrelation}
Figure~\ref{fig::PACF} shows that the autocorrelation of the sampled signal increases with sampling rate, consistent with Corollary~\ref{coro::autocorrelation}. This figure shows that sampling at low rate quickly destroys the innate autocorrelation of the signal, fundamentally altering the properties of the signal and rendering the prediction task harder.
We also see that the correlation between sampled ground truth  
%\note{$X$ or $Y$??} 
signal $Y$ and the external signal gradually increases in agreement with Corollary \ref{coro::covariance}.
% It is natural to expect that the knowledge of external signal will be less and less beneficial for prediction purpose as the percentage of information loss grows.
As a consequence of the loss of autocorrelation and correlation with the external signal, we can observe in Figure \ref{fig::RMSEvsDropout}, that at low sampling rates, the accuracy of the Poisson predictor is competitive to the ARIMA predictors.

\paragraph{Increase in Permutation Entropy}
\begin{figure}[H]
\centering
\includegraphics[width=0.7\columnwidth]{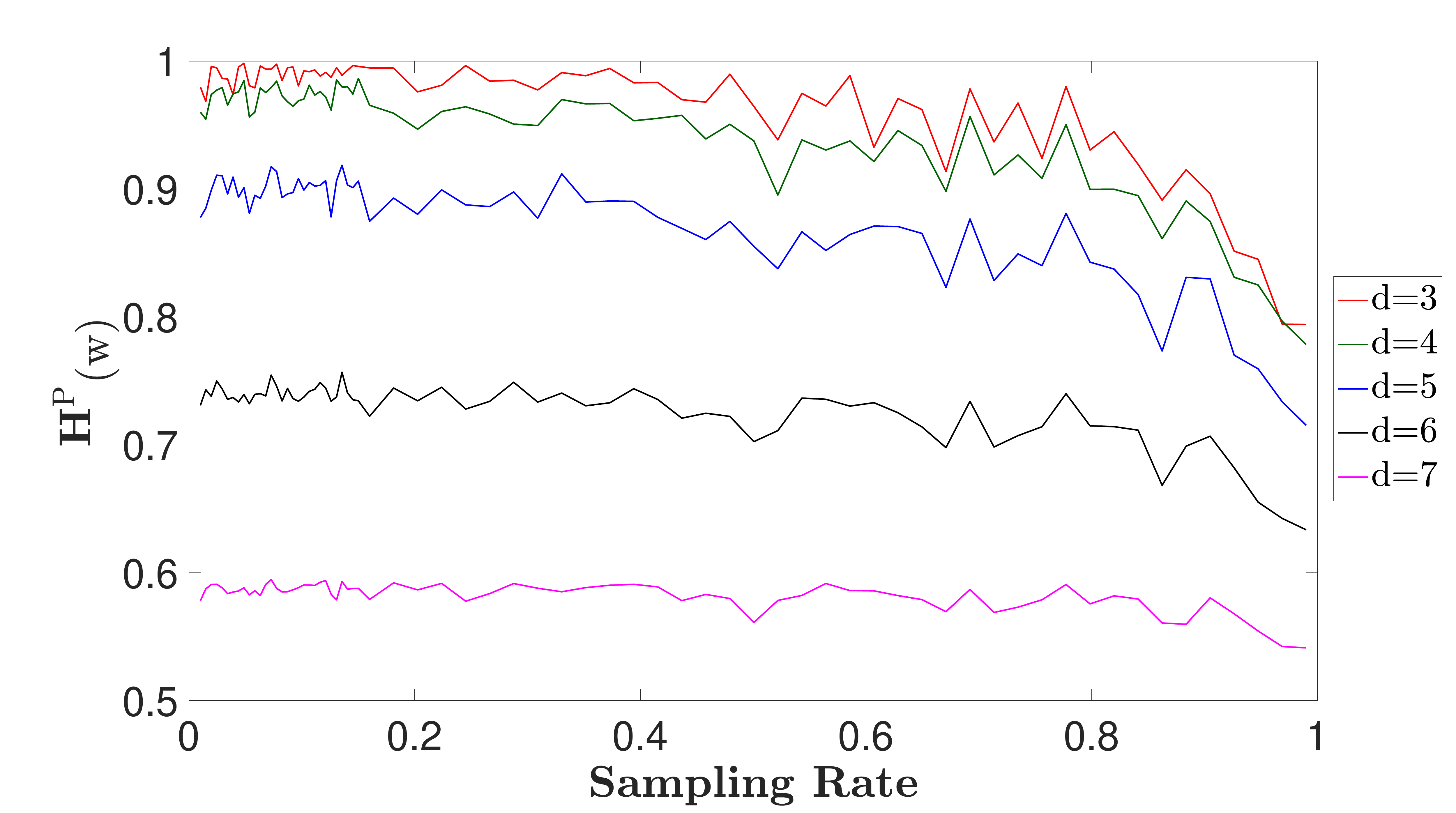}
\caption{Loss of predictability due to sampling for the synthetic data shown in Figure \ref{fig:signals}. Normalized weighted permutation entropy with respect to sampling rate increases at low sampling rates, showing the system becomes less predictable. Parameter $d$ is the embedded time dimension used to compute corresponding weighted permutation entropy $H^\mathrm{P}_{(\textrm{w})}$. The delay dimension is set to $\tau=1$.
%\note{AA: The second plot should compare the correlation between the external signal and ground truth.} \textcolor{blue}{addressed}
}
\label{fig::PermEntSyn}
\end{figure}

Figure \ref{fig::PermEntSyn} shows that the weighted permutation entropy decreases at high sampling rates, when more of the signal is retained. %, and levels off when sampling rate is around $85\%$. For the rest of sampling rate spectrum, $H^\mathrm{P}_{(\textrm{w})}$ decreases slowly and mostly stays stable.
This shows the loss of predictability of the process at low sampling rates. The trends in the figure imply that even when little of the signal is filtered out,  its predictability significantly degrades.

\paragraph{Decrease in Mutual Information}

\begin{figure}[H]
    \centering
    % \begin{subfigure}{.7\textwidth}
    %     \centering
    %     \includegraphics[width=\linewidth]{Figures/exogs.eps}
    % \end{subfigure}
    \begin{subfigure}{.7\textwidth}
        \centering
        \includegraphics[width=\linewidth]{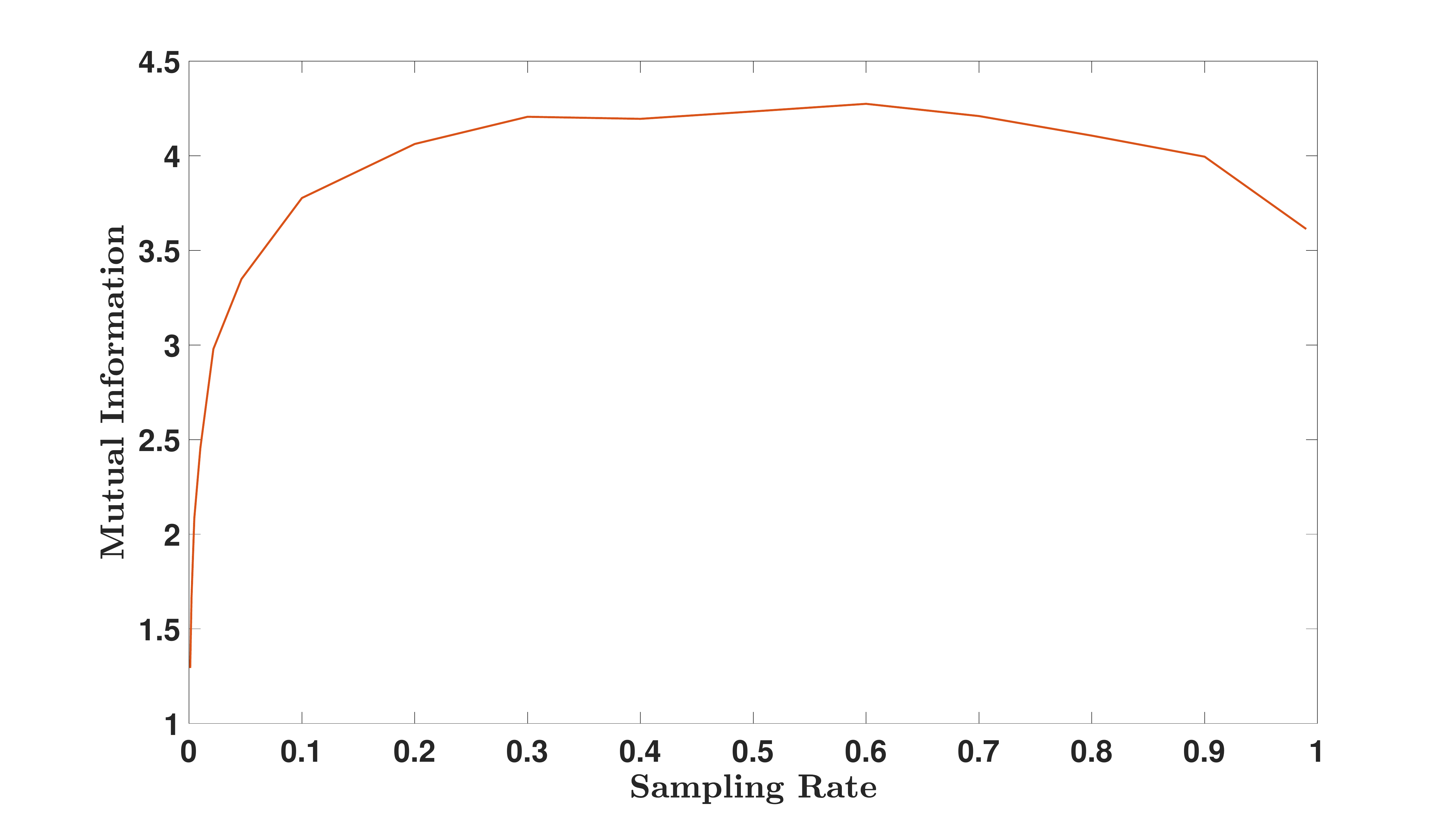}
    \end{subfigure}
    \caption{%(a) External Signal $S_t$ used to generate ground truth signal in figure \ref{fig:signals}. (b) 
    Decay of mutual information between the external signal and the sampled signal. Although the external signal in this illustration is highly correlated with the original (unfiltered) ground truth signal, sampling leads to a sharp loss of information about the original signal.}
        \label{fig::mi_decay}
\end{figure}

The loss of predictability cannot be offset using an informative external signal. This is because even if the external signal is highly correlated with the ground truth signal, sampling reduces its utility in predictive tasks. Figure \ref{fig::mi_decay} shows this decay in mutual information between the external and observed signals at low sampling rates.  The informative external signal does not reduce the uncertainty of the observed signal.

\subsubsection*{Nonstationarity of Prediction Errors}
\label{sec-si-garch}
Another quantity that characterizes the impact of sampling on the predictability of a time series is the covariance of the prediction error at different times, as the prediction error can intuitively reflect how well one can predict the event count. We show that sampling the time-series will introduce an autocorrelation into the prediction error and % at early time and the current time, and %renders the prediction error subject to the evolution of the counting process, thus resulting in the prediction error amplifying or decaying depending on the coefficients of an ARIMA time series model.
render it dependent on the evolution of the counting process, with errors growing larger or smaller depending on the type of process.
Next, we provide an example to motivate how sampling can induce a correlation between variances of predictions at different times.
\begin{proposition}\label{prop:garch}
Consider an auto-regressive (AR) process: $X_t = \alpha X_{t-1} + \varepsilon_t,$ with $\varepsilon_t$ white noise. The variance at the next step is given by,
\[
\mathrm{Var}(Y_{t}|X_{t}) = \alpha \mathrm{Var}(Y_{t-1}|X_{t-1}) + \varepsilon'_t.
\]
\end{proposition}
\begin{proof}
% Consider the following auto-regressive (AR) process: $$X = \alpha X_{t-1} + \varepsilon_t,$$ with $\varepsilon_t$ white noise.
The information filter can be described as a Binomial distribution, $Y\sim B(X,p)$. Hence, $\mathrm{Var}(Y_{t-1}|X_{t-1})=X_{t-1} p(1-p)$ using the fact that $Y_{t-1}|X_{t-1}$ is a Bernoulli random variable. Then, the variance at the next step is given by,
\begin{align*}
    \mathrm{Var}(Y_{t}|X_{t})=& X_{t} p(1-p)\\
    =& \left( \alpha X_{t-1} + \varepsilon_t\right) p(1-p)\\
    =&  \alpha X_{t-1}p(1-p) + \varepsilon'_t\\
     =&  \alpha \mathrm{Var}(Y_{t-1}|X_{t-1}) + \varepsilon'_t.
\end{align*}
\end{proof}
\noindent Proposition~\ref{prop:garch} shows that the variances of the sampled process, $Y$, are related by the exact same AR model that generated the process $X$. In other words, the conditional variance of the sampled process is autocorrelated. Notice, that this is not true for the unsampled data, given that $\mathrm{Var}(X_{t}|X_{t-1}) = \mathrm{Var}(\varepsilon_t)=\sigma^2$. %We generalize this result using Kalman filter theory.

Proposition~\ref{prop:garch}  shows that sampling a time series may introduce autoregressive conditional heteroskedasticity (ARCH) of the variance into the time series. This is usually tested by analyzing the residuals of the model. We use Engle's Langrage Multiplier test to demonstrate the appearance of ARCH effects in the residual signal.
Figure \ref{fig::statstest} shows that sampling does result in the introduction of ARCH effects to the residuals of predictions.

% We point out that the ARCH effect inferred by figure \ref{fig::statstest} is not contributed by the first 5 time lags that make up the original ground truth signal.

\begin{figure}[t]
\centering
\includegraphics[width=0.8\textwidth]{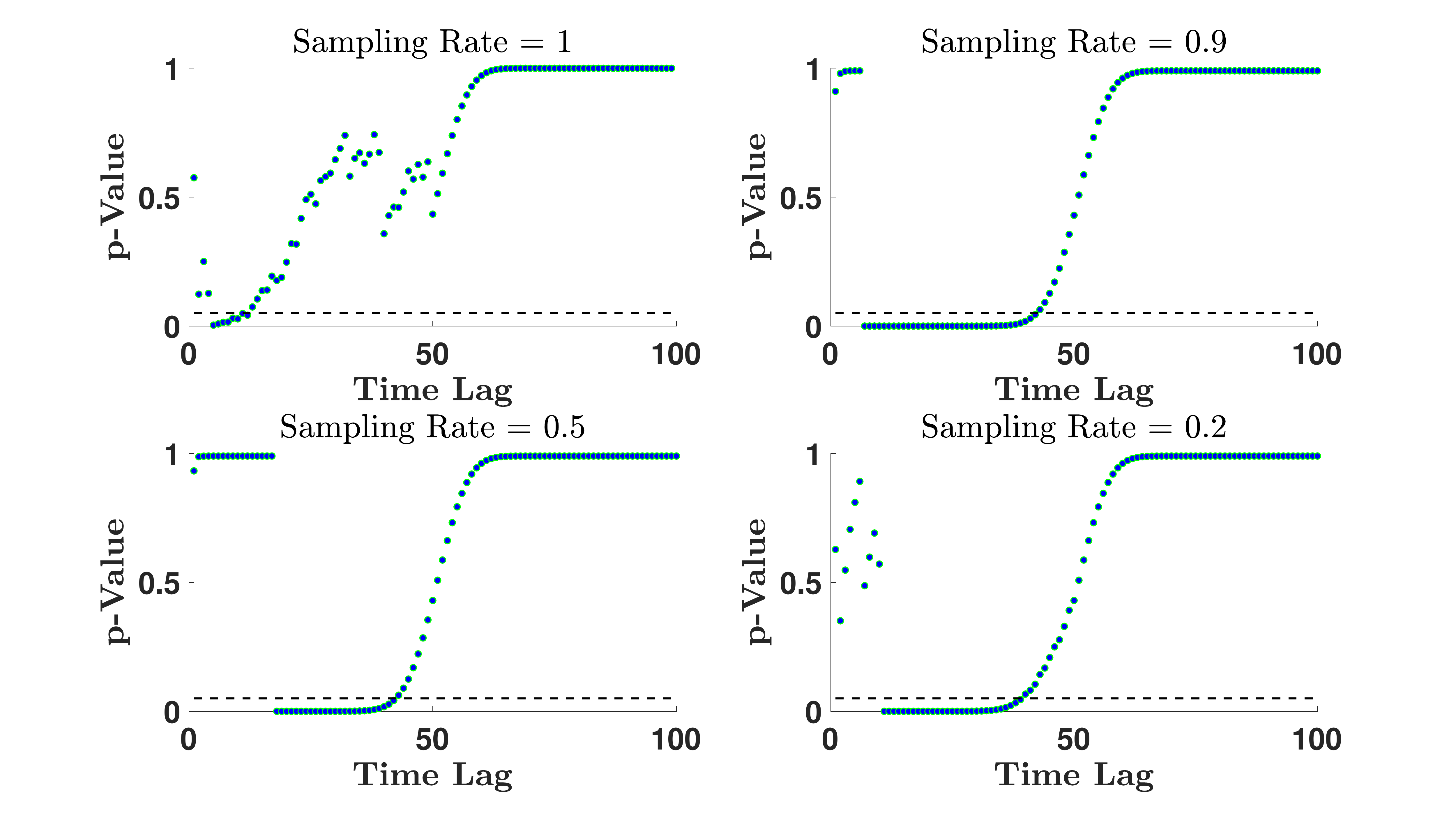}
\caption{$p$-Value versus time lag}
\label{fig::statstest}
\end{figure}

\begin{figure}[t]
\centering
\includegraphics[width=0.8\textwidth]{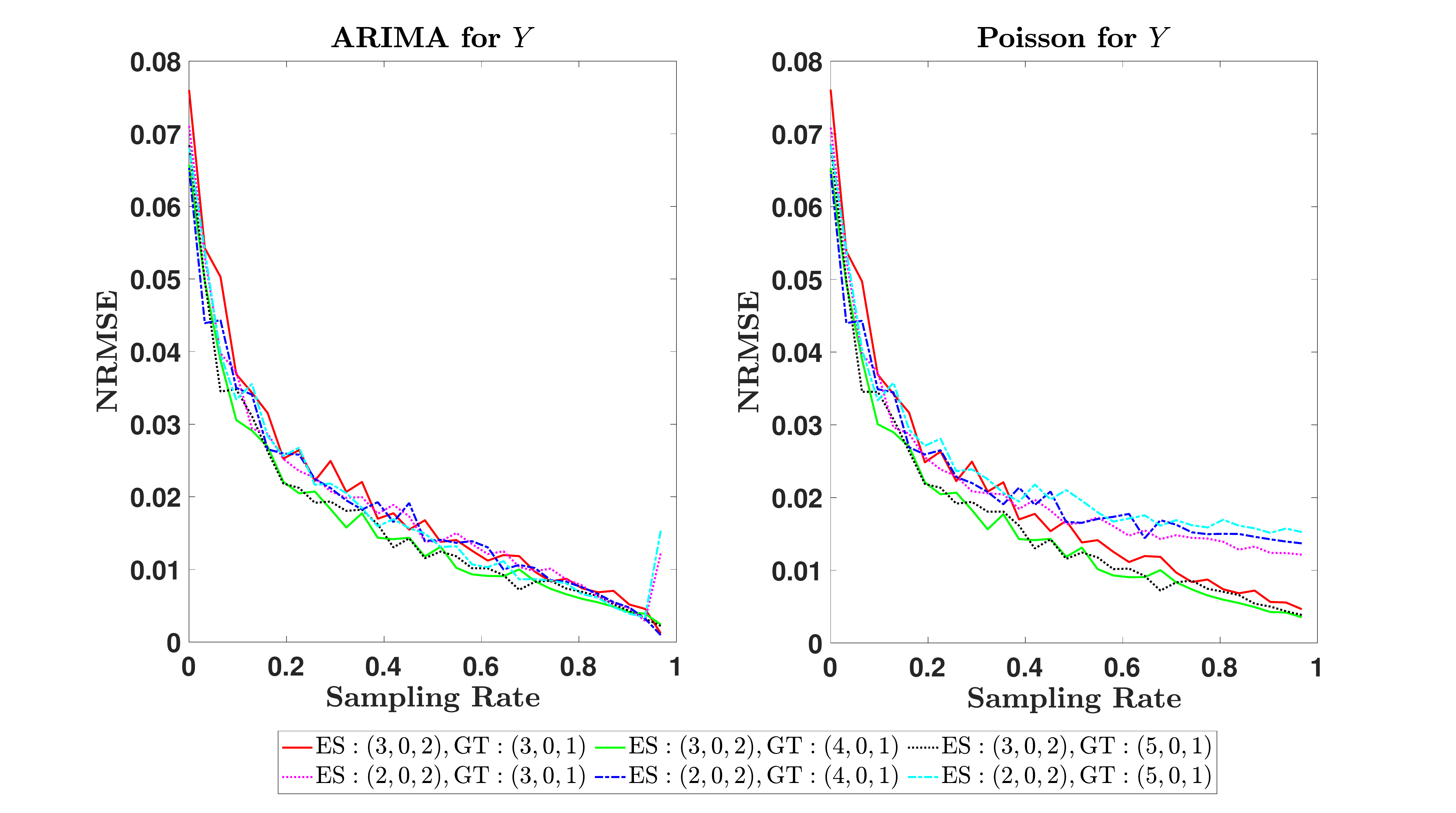}
\caption{Change of prediction accuracy (RMSE) with respect to sampling rate for various ground truth signals, when the sampled GTs are predicted without knowing external signal $S$.
%\note{KL: Also, make these square and put next to each other, like Fig 6.} \textcolor{blue}{addressed}
}
% \note{EF:Plots too small unreadable}}
\label{fig::multiARMA}
\end{figure}

We have applied Engle's Lagrange Multiplier (ELM) test \cite{engle1982autoregressive} on the residual $\hat{Y} - Y$, which is evaluated at time $t,t-1,\cdots,t-L$ where $L=100$ is the maximum time lag we consider, to examine the existence of ARCH behavior. Under the null hypothesis that there is no ARCH effect, the test statistic used in ELM test has the asymptotic distribution $\chi^2(L)$. When $p$-value is less than the significance level $\tilde{\alpha} = 0.05$, the ELM test says that we can reject the null hypothesis with 95\% confidence.

From Figure \ref{fig::statstest}, we see that for the original time-series (top-left plot) there is no heteroskedasticity, $p$-values are largely over the significant level threshold $0.05$ for most of time lags under consideration. In the case of $q=0$, ELM test does not provide us enough evidence to reject the hypothesis that ARCH effect does not exist. However, when we filter the original signal, as we see in the case $q=0.1,0.5,0.8$ respectively, the $p$-value for time lags from 6 to 40 are mostly below the significance level bar, which, by ELM test, strongly suggests that the null hypothesis be rejected. In other words, ELM test suggests with 95\% confidence that ARCH effect exists in the residual of prediction records on the filtered time-series.

% This numerical experiment supports the intuitive inference that as the autocorrelation is destroyed by information loss, the prediction error sequence takes on more autocorrelation in contrast to the scenario where there is no information loss, where prediction noise is by large a white-noise-like sequence which does not contain temporal autocorrelation.

% \vspace{.2in}
% \textcolor{red}{The following section is new, and the graph is updated to correct the $y$-scale.}

\subsubsection*{Generalizability}
\label{sec-si-generalizability}
So far, we have explored methodically one example where we have validated our theoretical results as well as provided new insights about the unpredictability of sampled time series. Here we show the average loss of predictability across many randomly generated ground truth time series.

We generate multiple external signal--ground truth signal pairs, and apply ARIMA and Poisson models on these data sets to test if the tendencies we have observed previously (i.e. nonlinear decrease of NRMSE with respect to sampling rate) will also appear in signals of different ARIMA orders.
We observe in Figure \ref{fig::multiARMA} similar behavior across multiple data sets:
\begin{enumerate}
    \item Increasing %tendency, increasing speed as well as the amplitude of NRMSE are by large uniform across all data sets despite small discrepancy in the low end of dropout spectrum.
prediction error (NMRSE) despite small discrepancy at high sampling rates. This general behavior is partially due to the normalization by sample mean when we evaluate the NRMSE;
\item The autoregressive model outperforms the non-autoregressive Poisson model for higher sampling rates, but, this difference fades out for lower sampling rates as a consequence of sampling. %the information loss.
\end{enumerate}

We further postulate that this decreasing tendency between increasing sampling rate and NRMSE bears generality with respect to any combination of ARIMA orders of external signal--ground truth data pair.
When predicting $Y$ with the external signal $S$, we observe a very similar tendency of NRMSE---sampling rate relationship as in Figure~\ref{fig::multiARMA}. Albeit noticeable difference in the high end of sampling rate spectrum, the decreasing tendency, decreasing speed and the amplitude of NRMSE are by large the same as in the not-using-$S$ case.

In the general setting of ARIMA model, we show a decreasing prediction error w.r.t.\ the sampling rate, and the shrinking of the advantage of autoregressive model in prediction accuracy over Poisson model as sampling rate becomes smaller. 

We demonstrated the decoupling of the external signal and the filtered signal as a prevalent phenomenon in the generic ARIMA setting. We observe in Figure \ref{fig:mi_general} that for all ground truth-external signal pairs, the lower the sampling rate is, the smaller mutual information becomes. Therefore, %for any given ARIMA orders of external signal-ground truth pairs, 
we postulate that the positive correlation between sampling rate and mutual information can always be observed, thus implying that the knowledge of external signal cannot help recover the predictability of the ground truth signal.
\begin{figure}[H]
\begin{center}
\centering
\includegraphics[width=0.8\columnwidth]{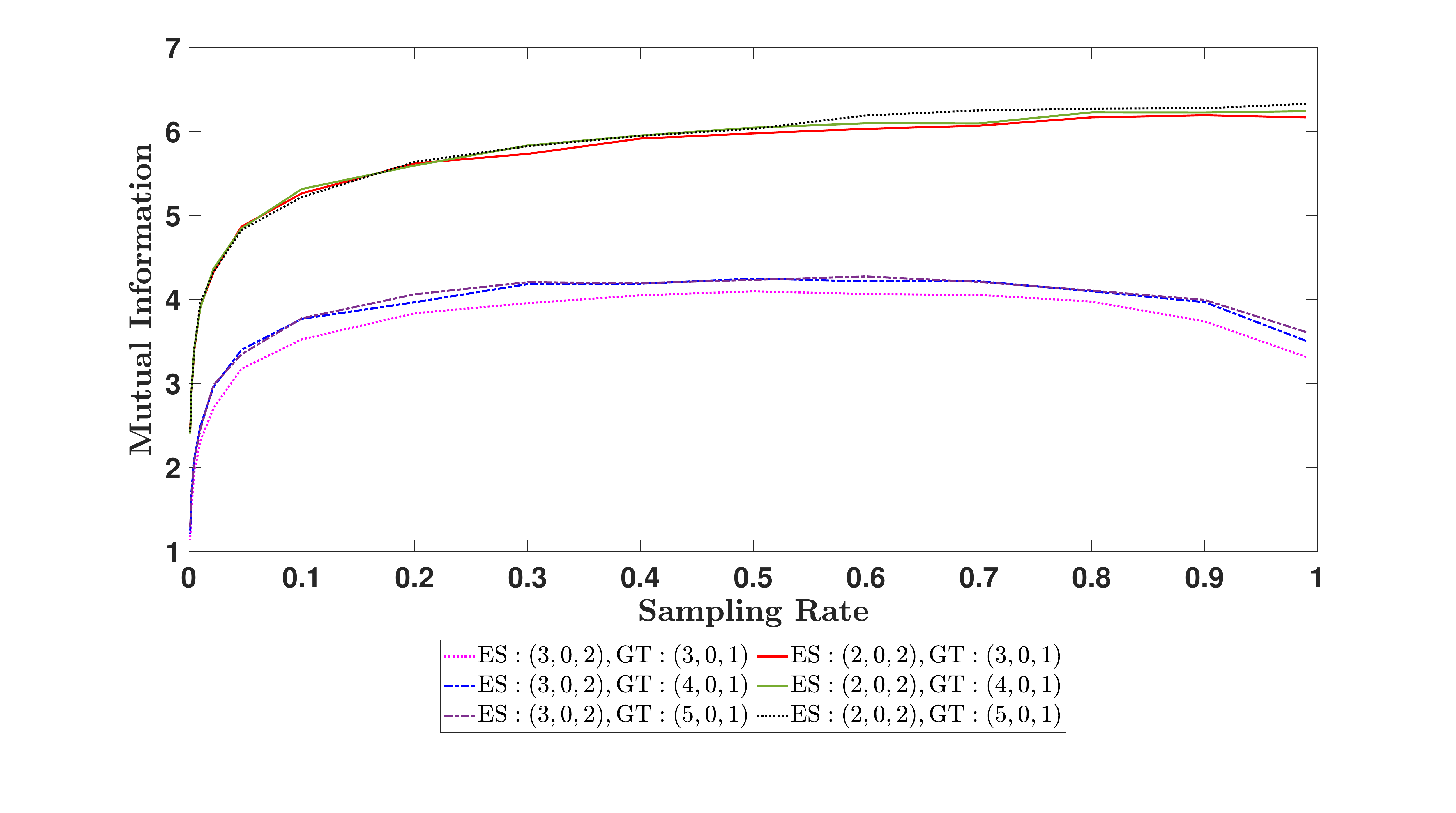}
\hspace{0.2cm}

\caption{Change of mutual information (MI) with respect to sampling rate for various (ground truth, external signal) pairs.}
\label{fig:mi_general}
\end{center}
\end{figure}

%%% EMPIRICAL_SI
\subsection*{Supplementary Figures}
\subsubsection*{Epidemics.}
\begin{figure}[H]
\begin{center}
\centering
\includegraphics[width=0.55\columnwidth]{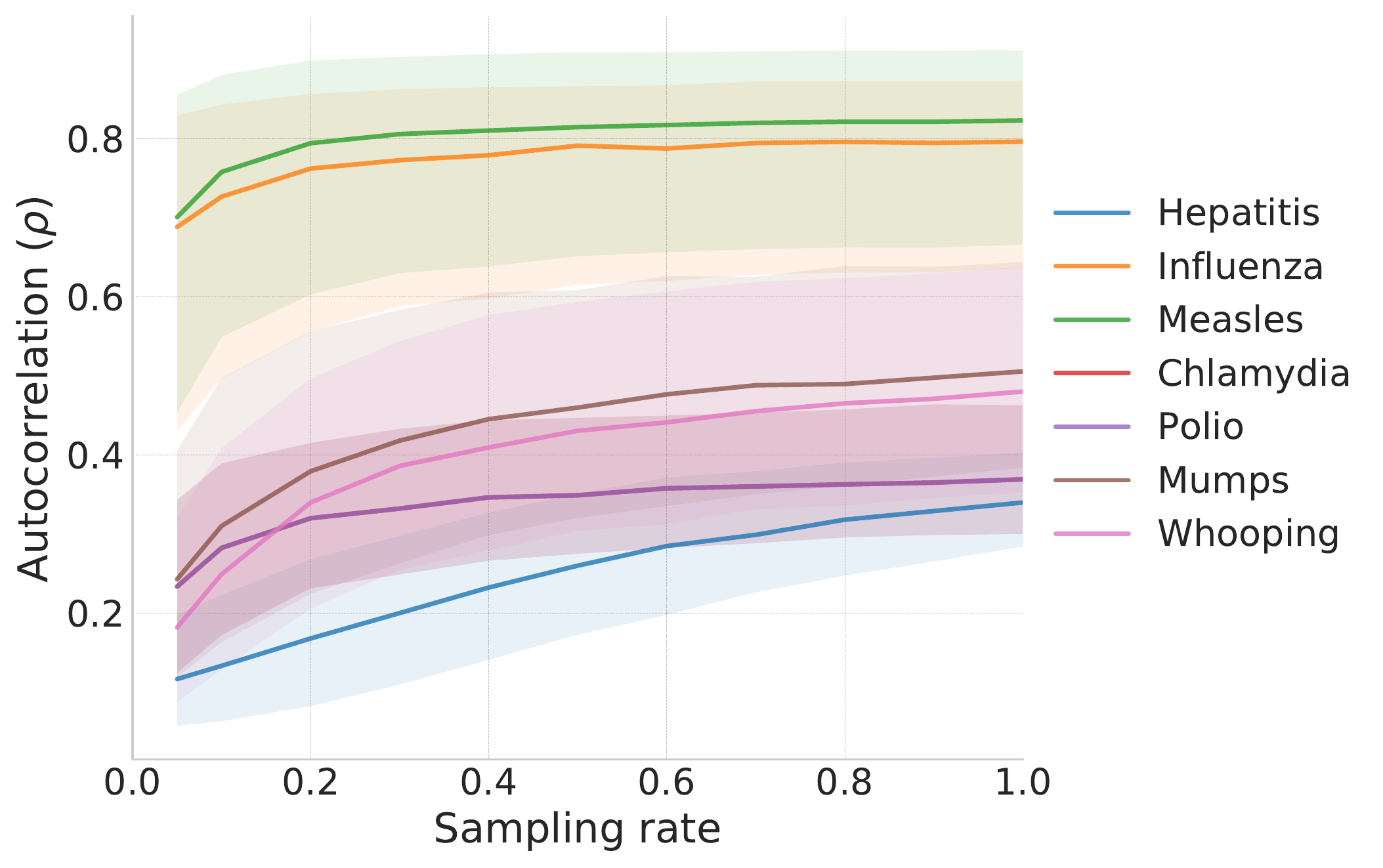}
\hspace{0.3cm}

\caption{Theoretical and empirical Loss of autocorrelation in outbreaks due to sampling for all diseases.   The plot depicts a decrease on the autocorrelation as drop-out rate increase. For each of the eight weekly, state-level diseases, we selected $100$ random points and calculated the entropy and autocorrelation for different drop-out rates over a one year window. The solid line represents the median, shaded region marks the interquartile range.}
\label{fig:epidemics_theo_si}
\end{center}
\end{figure}

\begin{figure}[H]
\begin{center}
\centering
\includegraphics[width=0.55\columnwidth]{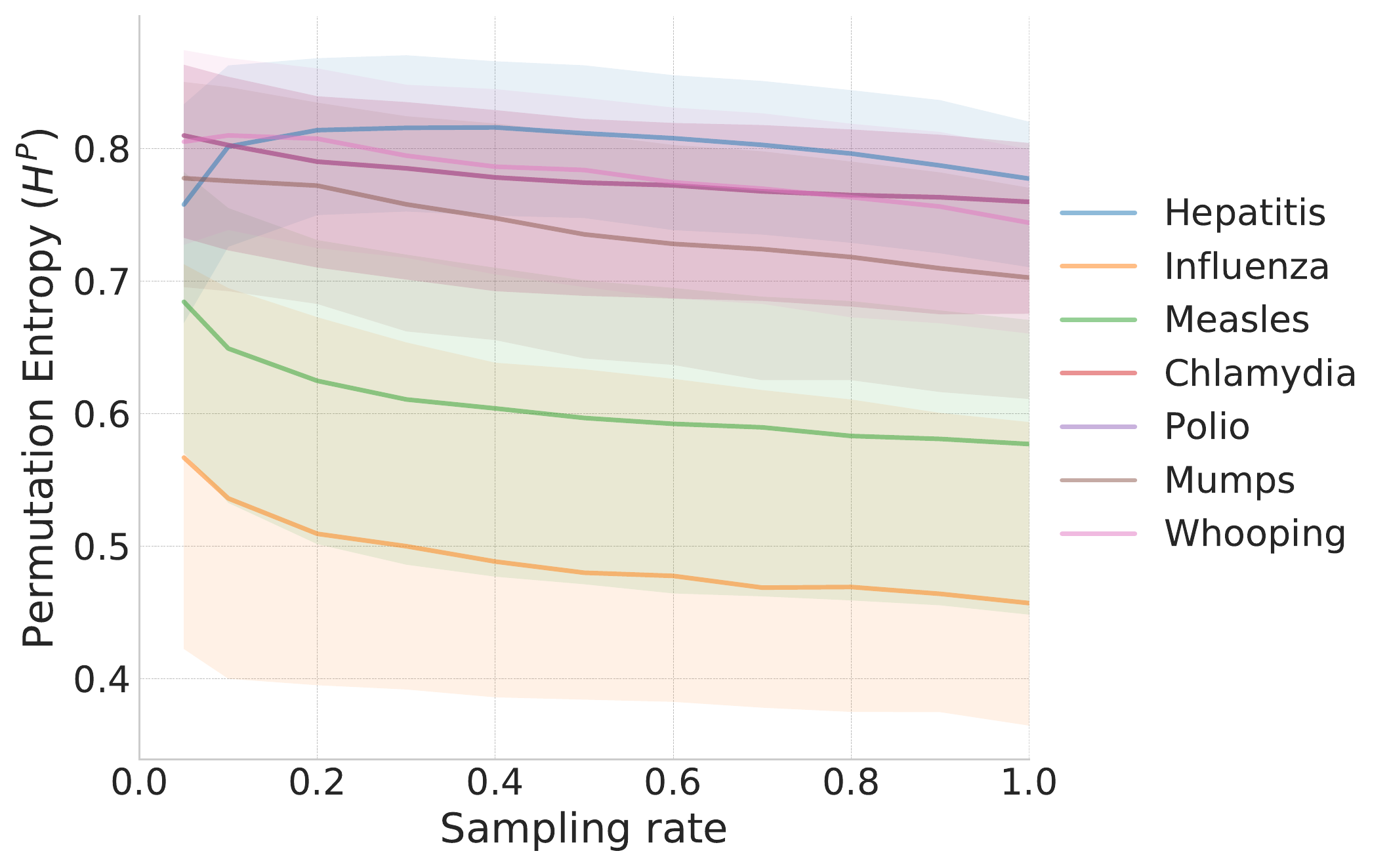}
\caption{Loss of predictability due to sampling for all diseases.   
The plot shows and increase on the permutation entropy as drop-out rate increase. For each of the eight weekly, state-level diseases, we selected $100$ random points and calculated the entropy and autocorrelation for different drop-out rates over a  one year window. The solid line represents the median, shaded region marks the interquartile range. }
\label{fig:exp_epidemics_si}
\end{center}
\end{figure}

\begin{figure}[H]
\begin{center}
\centering
\includegraphics[width=0.55\columnwidth]{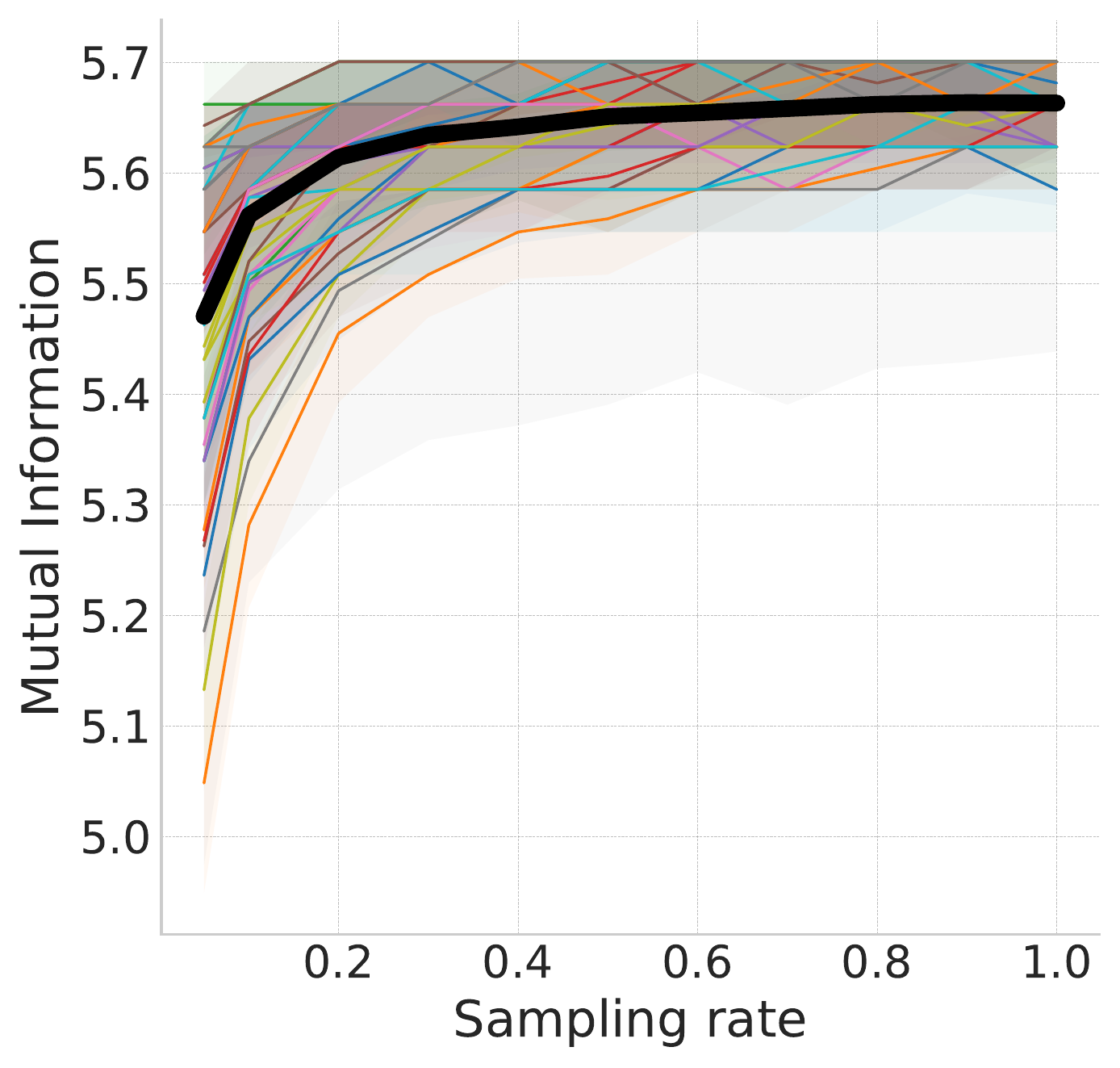}
\caption{Decay of mutual information with external signal due to sampling of the influenza activity.   
For each state, we selected $100$ random one-year time windows and calculated the median mutual information~\cite{pyinform_asu} between Google Flu trends and the influenza activity at different sampling rates. Shaded regions mark the inter-quartile ranges for each state; the solid line represents the average coefficient across all states.}
\label{fig:exp_epidemics_mi}
\end{center}
\end{figure}

\begin{figure}[H]
\begin{center}
\includegraphics[width=0.49\linewidth]{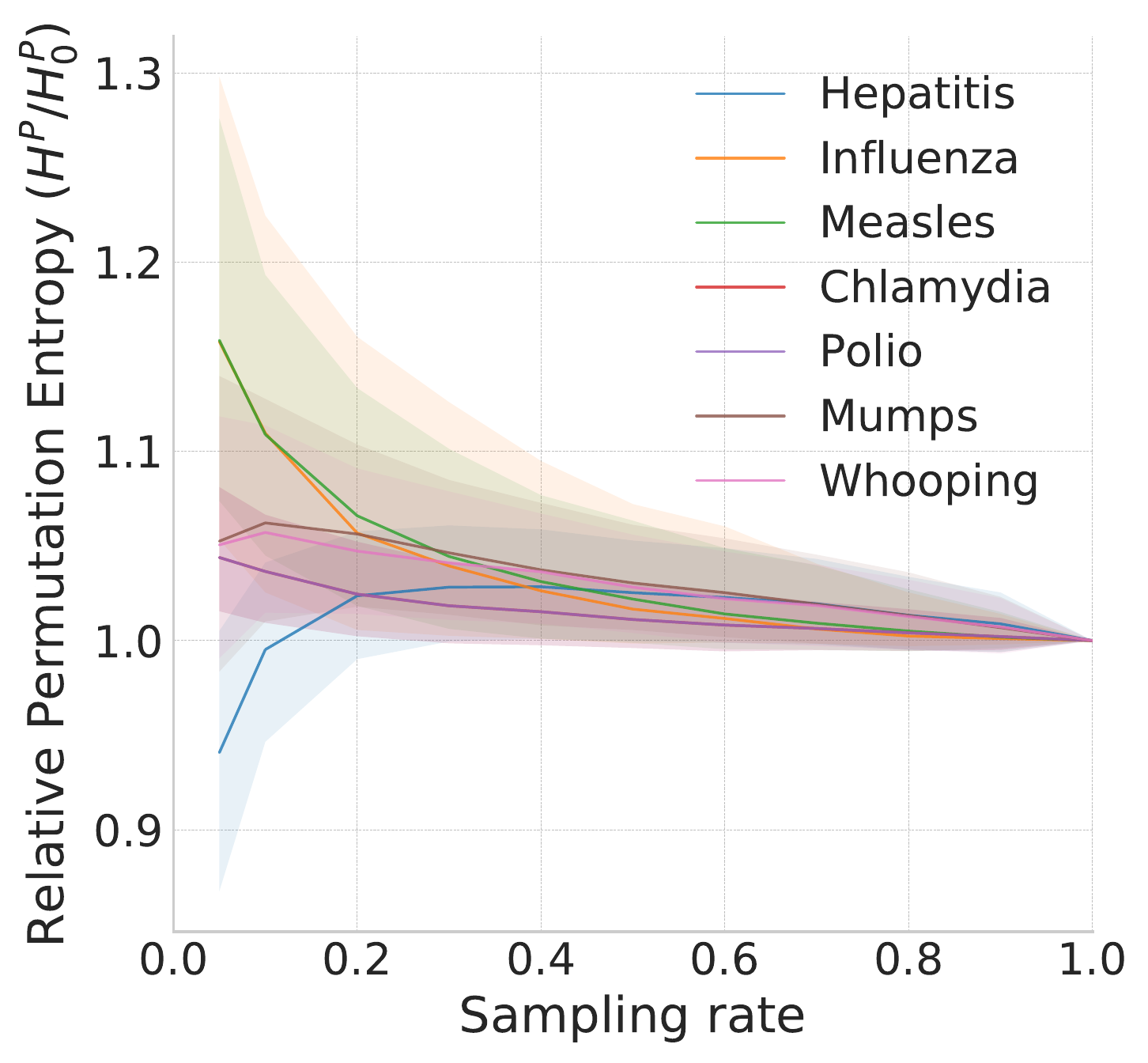}
% \hspace{0.1cm}
\includegraphics[width=0.49\linewidth]{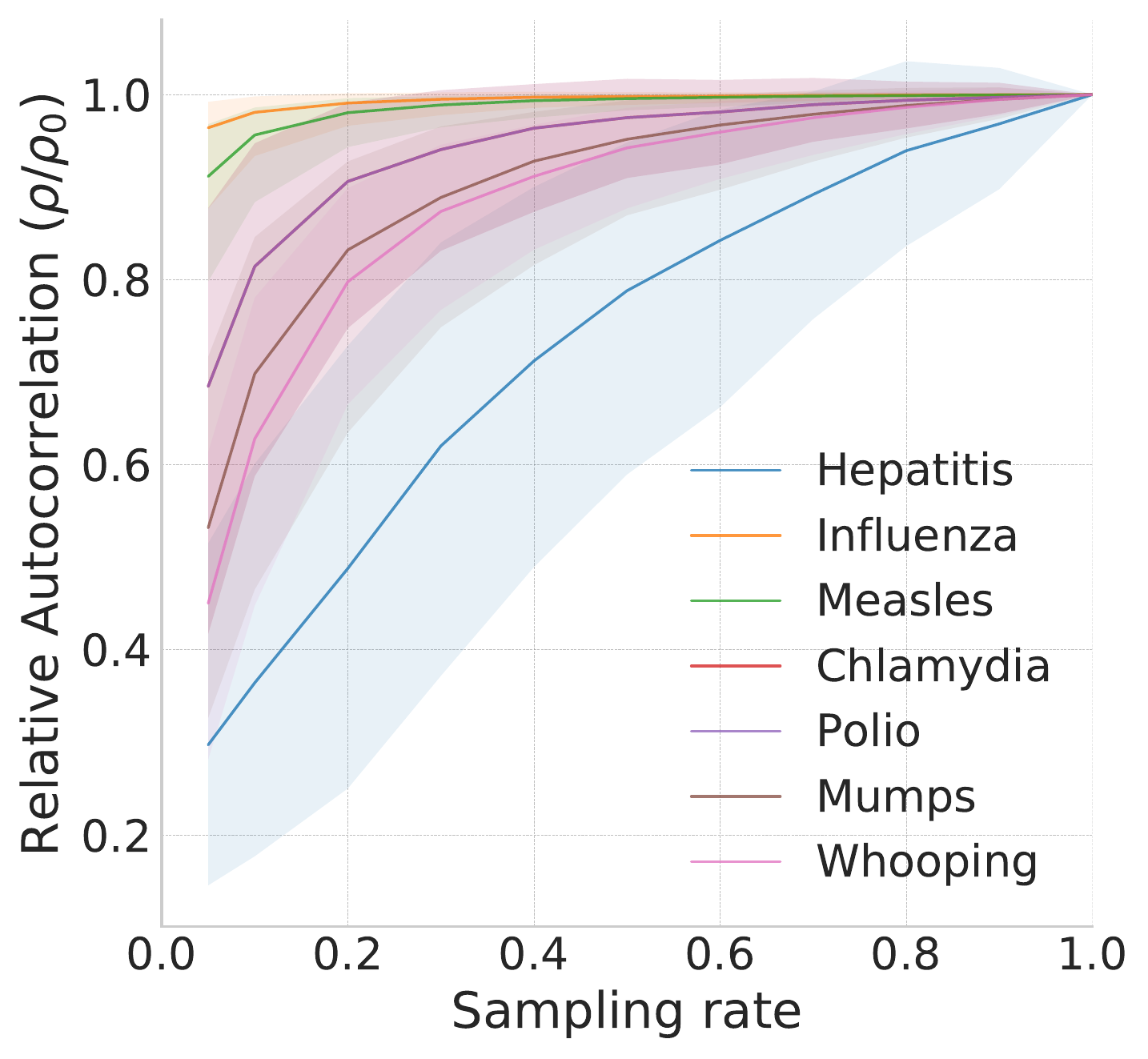} \\
% \hspace{0.1cm}
\includegraphics[width=0.49\linewidth]{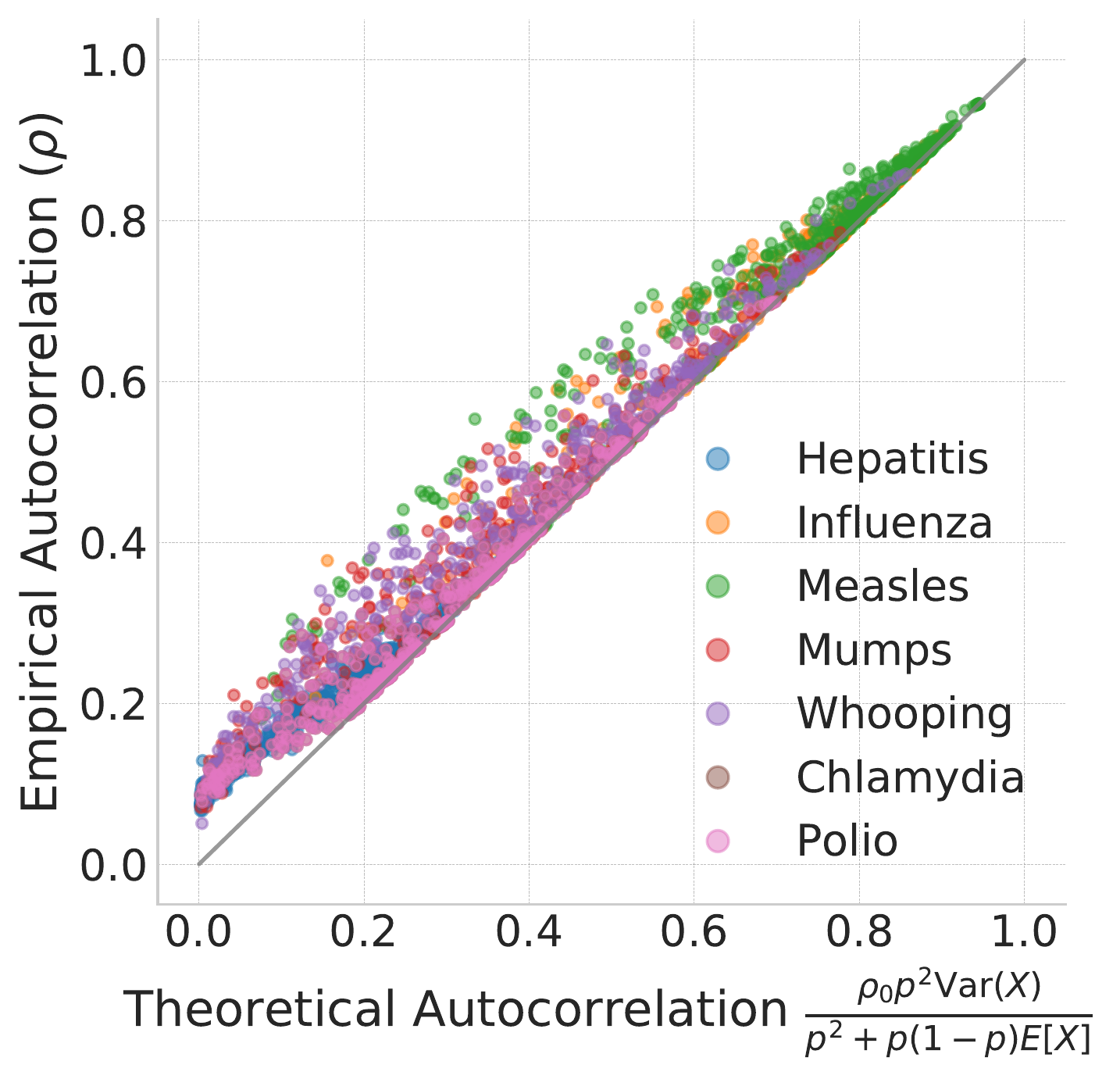}
\caption{Loss of predictability of disease outbreaks due to sampling.   The plots show a decrease in permutation entropy (\textbf{top-left}) and an increase in autocorrelation (\textbf{top-right}) of the outbreak time series for increasing sampling rates. For each of the eight %weekly, state-level 
diseases, we selected $100$ random two-year time windows and calculated the relative weighted permutation entropy  and autocorrelation for different sampling rates over that window. The solid line represents the median ratio across all states between the original time series and the sampled one; shaded regions mark the inter-quartile ranges.  The \textbf{bottom} plot supports our theoretical results by plotting Equation~\ref{eq:autocorr}  against the empirical autocorrelation of the sampled time series at different sampling rates for each disease. } %state-level disease.}
\label{fig:exp_epidemics_w104}
\end{center}
\end{figure}

\newpage
\subsubsection*{Social Media.}

\begin{figure}[H]
\begin{center}
\centering
\includegraphics[width=0.45\columnwidth]{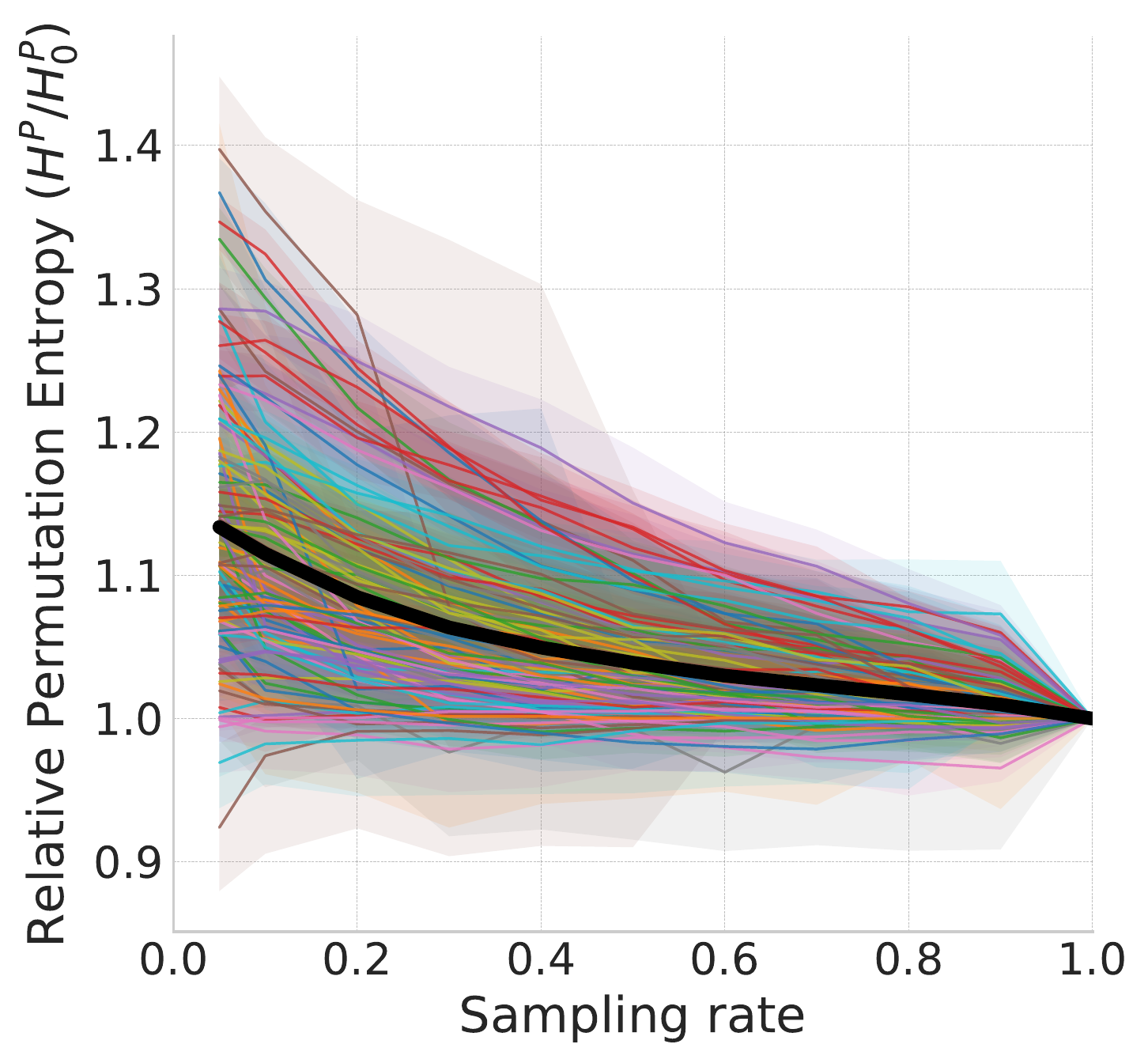}
\includegraphics[width=1.\columnwidth]{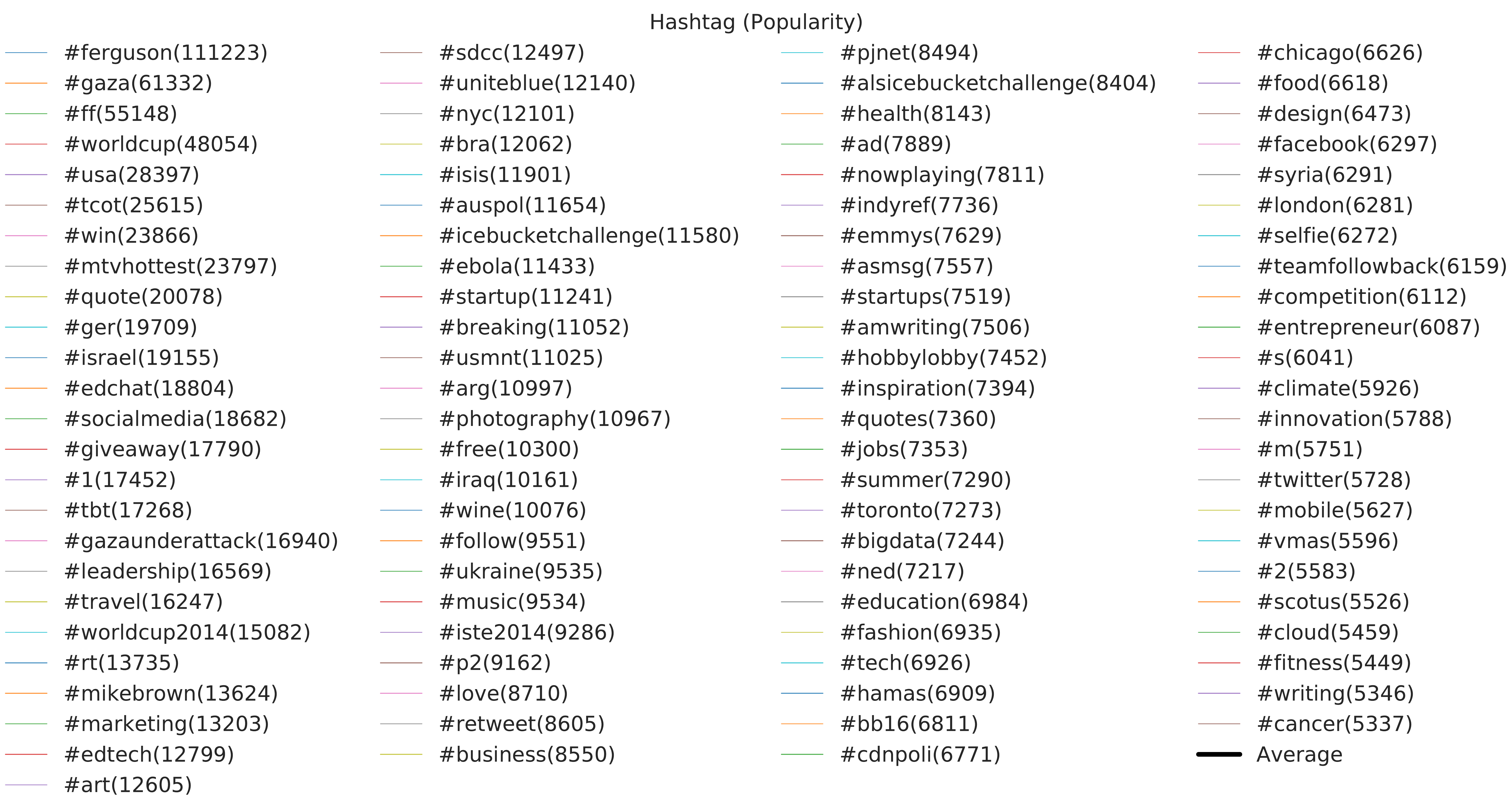}
\caption{Loss of predictability in social media due to sampling on user's activity (posts made by the user). The plot shows the median weighted permutation entropy relative to the original time series for each of the top 100 most popular hashtags. }
\label{fig:exp_twitter_entropy1}
\end{center}
\end{figure}

\begin{figure}[H]
\begin{center}
\centering
\includegraphics[width=0.45\columnwidth]{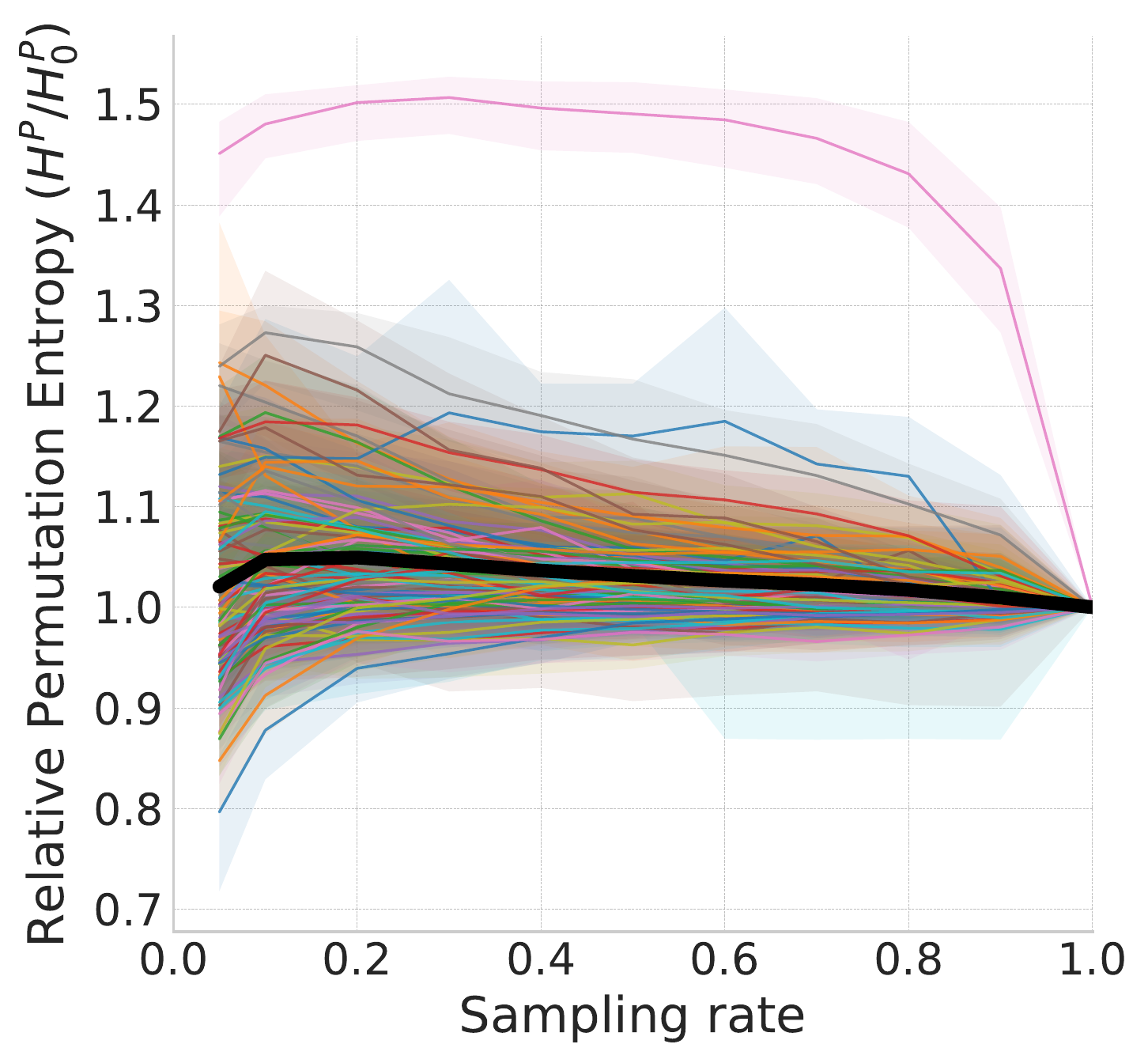}
\caption{Loss of predictability in social media due to sampling on user's activity (posts made by the user). The plot shows the median weighted permutation entropy relative to the original time series for each of the top 100 most active users.}
\label{fig:exp_twitter_entropy2}
\end{center}
\end{figure}

\begin{figure}[H]
\begin{center}
\centering
\includegraphics[width=0.8\columnwidth]{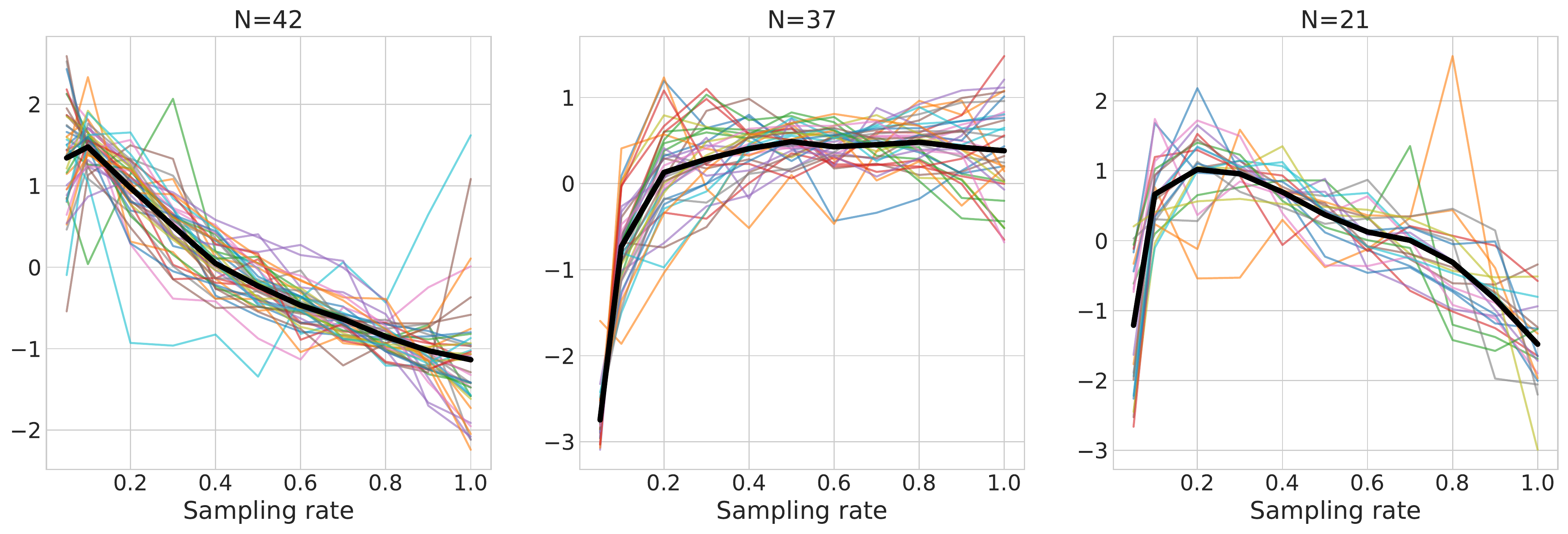}
\includegraphics[width=0.8\columnwidth]{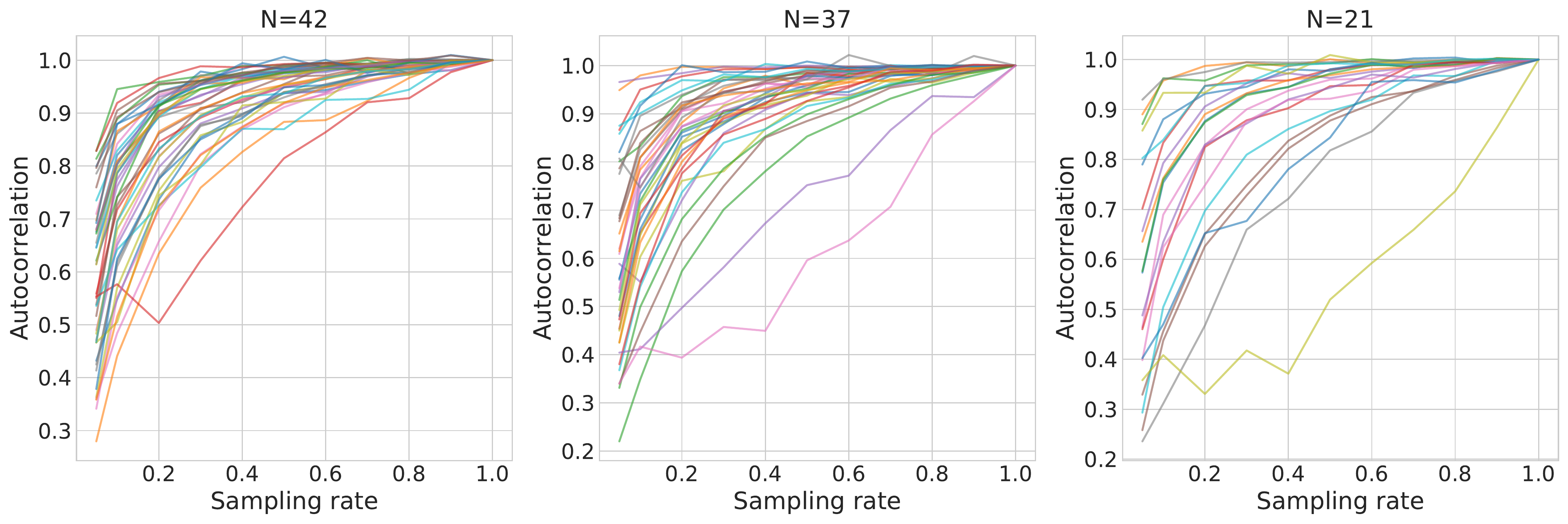}
\caption{Clustering the permutation entropy behavior from Figure~\ref{fig:exp_twitter_entropy2}. We use, 
K-means clustering to depict the three most characteristics types of behaviors exhibited when computing predictability in social media as a function of sampling on user's activity.}
\label{fig:exp_twitter_entropy_cluster}
\end{center}
\end{figure}

\begin{figure}[H]
\begin{center}
\centering
\includegraphics[width=0.4\columnwidth]{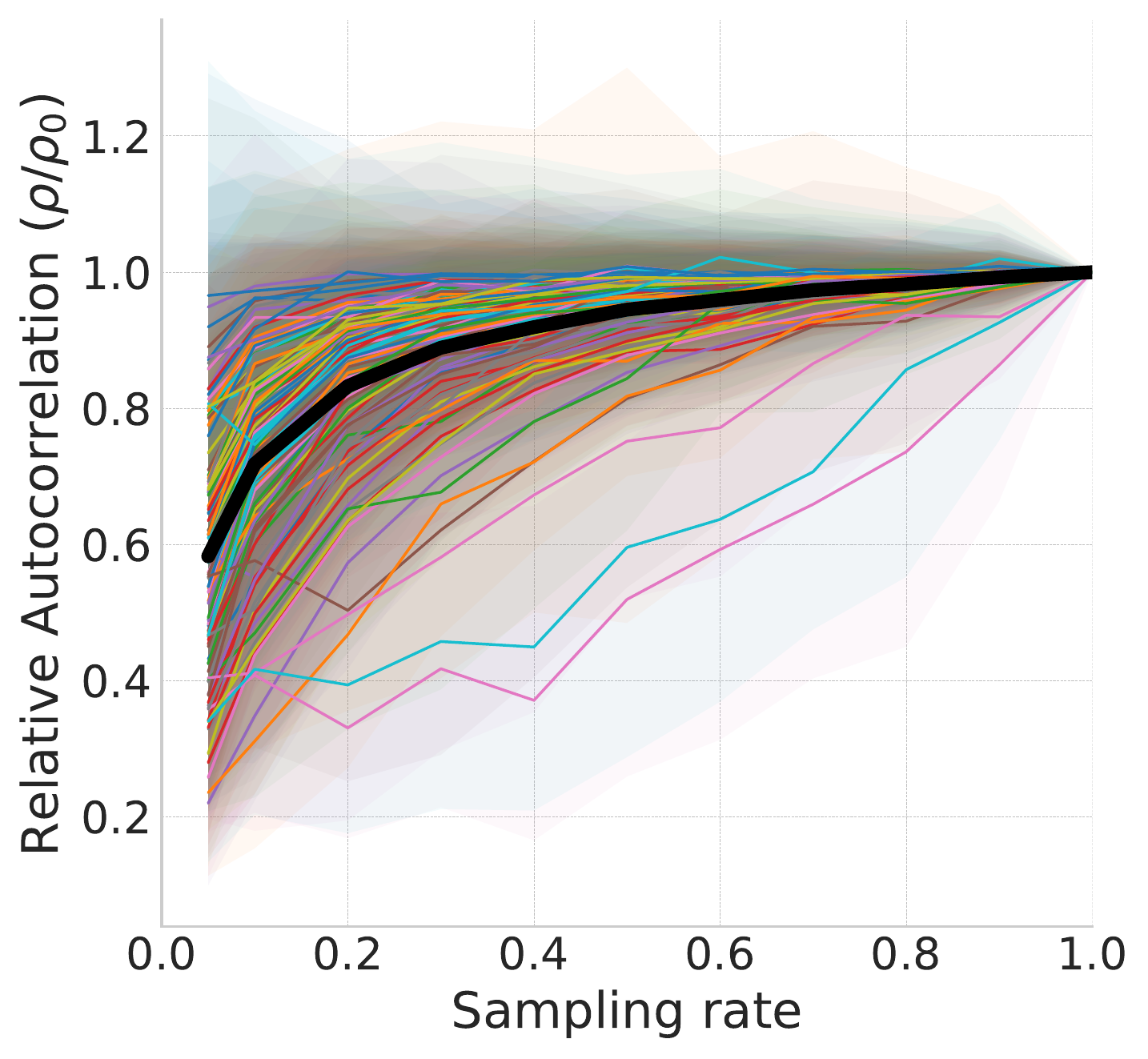}
\hspace{0.3cm}
\includegraphics[width=0.4\columnwidth]{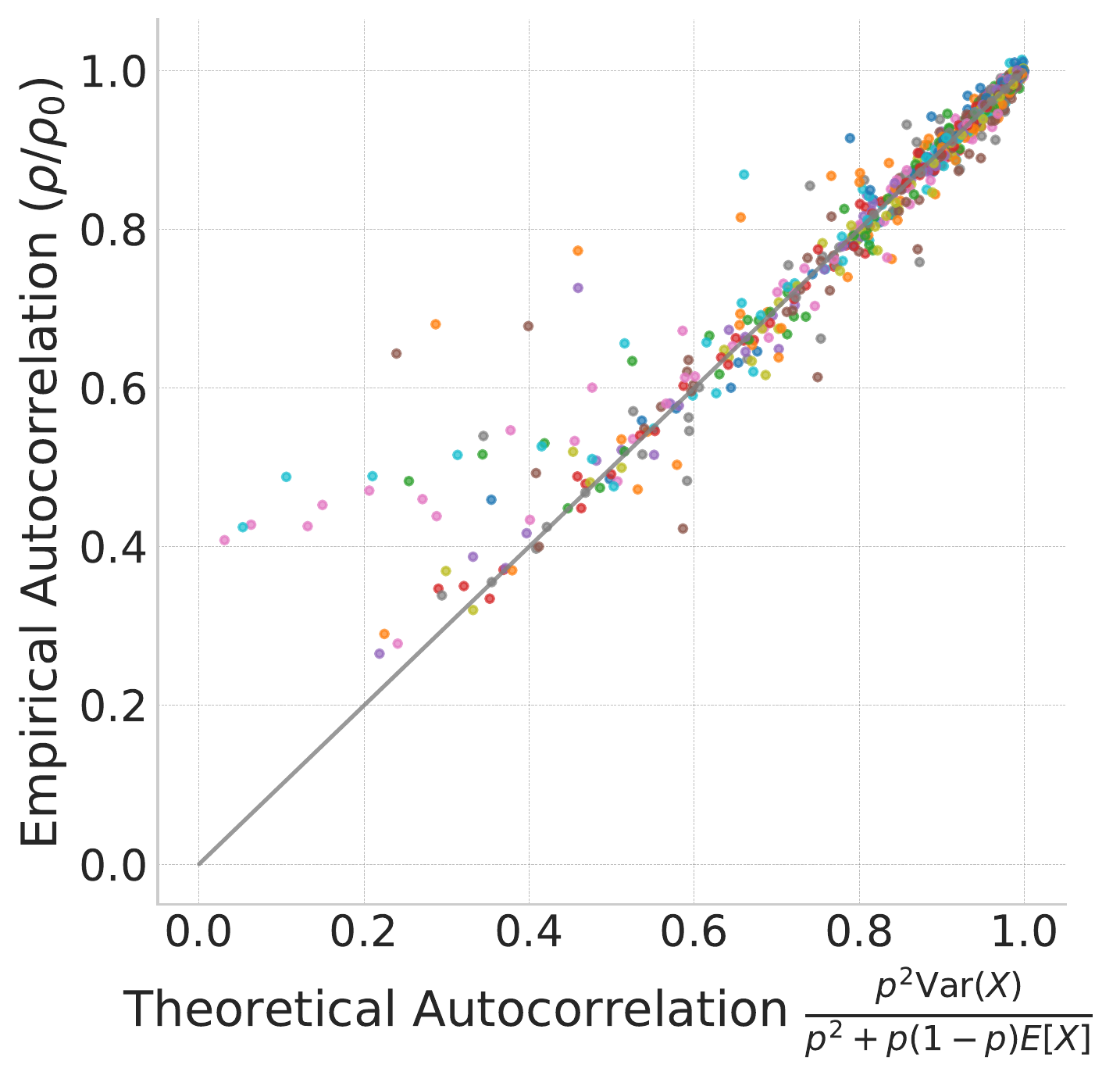}
\caption{Empirical and theoretical effects of sampling on autocorrelation of user's activity. The left plot shows the median autocorrelation relative to the original time series for each of the top 100 most active users; shaded region marks the interquartile ranges; the black line represents the average autocorrelation across all users.
The right plot shows the accuracy of the theoretical prediction according to Equation \ref{eq:autocorr}. The  line depicts the identity function to represent an accurate fit to the data. }
\label{fig:exp_user}
\end{center}
\end{figure}

% \newpage
\subsection*{Cryptocurrencies.}
% \begin{center}
% \begin{figure}[h]
% \centering
% \includegraphics[width=0.5\columnwidth]{Figures/crypto_cov.pdf}
% \caption{Decay of covariance between cryptocurrencies repositories popularity and their prices for different drop-out rates. Each point is the median Pearson's correlation coefficient over $1000$ samples. Error bars show the standard deviation; Y-axis is in log scale.
% For each cryptocurrency, we calculated over $1000$ samples, the median  normalized covariance $\frac{\mathrm{Cov}(Y,S)}{\mathrm{Cov}(X,S)}$ (\textbf{left}) and the Pearson's correlation coefficient (\textbf{right}) between the price and the popularity of related Github repositories at various drop-out rates. Shaded region marks the interquartile ranges for each coin.}
% \label{fig:crypto_cov}
% \end{figure}
% \end{center}

\begin{center}
\begin{figure}[h]
\centering
\includegraphics[width=\columnwidth]{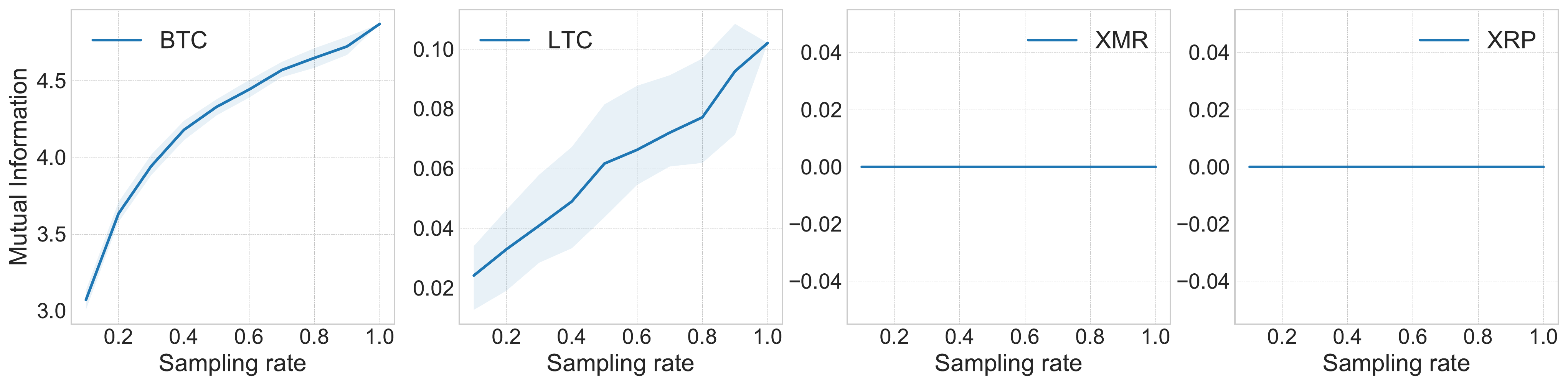}
\caption{Decay of mutual information between the popularity of cryptocurrencies repositories and their prices for different sampling rates.
For each cryptocurrency, and each sampling rate, we obtained $100$ samples, and calculated the median mutual information~\cite{pyinform_asu} between the price and the popularity of related Github repositories. The solid line represents the median mutual information for each coin. Shaded region marks the interquartile ranges for each coin.}
\label{fig:crypto_mi}
\end{figure}
\end{center}

\newpage
\subsection*{Supplementary Tables}

\begin{table}[h!]
\begin{center}
\begin{tabular}{ccccc}
\toprule 
Disease & Years & Total Infections & Infections/Year (SD) & Infections/week (SD) \tabularnewline
\midrule
Hepatitis & 1966-2014 & $742,554$ & $15,154$ ($18,905$)  & $14,011$ ($2,643$)  \tabularnewline
Influenza & 1919-1951 & $6,498,817$ & $196,934$ ($198,551$)  & $122,619$ ($163,348$)  \tabularnewline
Measles & 1909-2001 & $18,430,036$ & $198,172$ ($247,833$)  & $347,737$ ($310,744$)  \tabularnewline
Chlamydia & 2006-2014 & $4,882,110$ & $542,456$ ($80,216$)  & $92,115$ ($18,487$)  \tabularnewline
Polio & 1921-1971 & $505,246$ & $9,907$ ($13,134$)  & $9,532$ ($10,676$)  \tabularnewline
Gonorrhea & 1972-2014 & $3,701,913$ & $86,091$ ($210,930$)  & $69,847$ ($13,426$)  \tabularnewline
Mumps & 1967-2014 & $866,965$ & $18,061$ ($35,054$)  & $16,357$ ($9,180$)  \tabularnewline
Whooping cough & 1909-2014 & $2,220,008$ & $20,943$ ($47,835$)  & $41,887$ ($7,296$)  \tabularnewline
\bottomrule
\end{tabular}
\par\end{center}
\caption{Epidemics data: descriptive statistics.}
\label{si-epidemics}
\end{table}

\begin{table}[h!]
\begin{center}
\begin{tabular}{cccccc}
\toprule 
Hashtags & Users & Tweets & Tweets/hashtag (SD) & Tweets/User (SD) & Tweets/Day (SD)\tabularnewline
\midrule
$100$ & $233,108$ & $1,269,348$ & $12,693$ ($13,561$) & $5.4$ ($12$) & $11,135$ ($6,941$)\tabularnewline
\bottomrule
\end{tabular}
\par\end{center}
\caption{Twitter hashtag activity descriptive statistics. Top 100 most popular hashtags.}
\label{si-twitter-S1}
\end{table}

\begin{table}[h!]
\begin{center}
\begin{tabular}{cccccc}
\toprule 
Users & Tweets & Hashtags & Tweets/User (SD) & Tweets/hashtag (SD) & Tweets/Day (SD) \tabularnewline
\midrule
$150$ & $167,654$ & $20,990$ & $1,118$ ($13,561$) & $8.0$ ($41.6$) & $1,552$ ($1,668$)  \tabularnewline
\bottomrule
\end{tabular}
\par\end{center}
\caption{Twitter user activity descriptive statistics. Top 150 most active users.}
\label{si-twitter-S2}
\end{table}

\begin{table}[h!]
\begin{center}
\begin{tabular}{cccccc}
\toprule 
Cryptocurrency & Events & Repositories & Users & Events/Day (SD) & Events/Repo (SD)  \tabularnewline
\midrule
Bitcoin (BTC) & $40,038$ & $1,962$ & $5,324$ & $460$ ($116$)  & $20.4$ ($138.4$) \tabularnewline
Ripple (XRP) & $2,963$ & $7$ & $86$ & $35.3$ ($27.7$)  & $423.2$ ($1,115$) \tabularnewline
Litecoin (LTC) & $1,222$ & $137$ & $302$ & $14$ ($11$) & $8.9$ ($19.2$)  \tabularnewline
Monero (XMR) & $370$ & $15$ & $54$ & $5.7$ ($6.3$) & $24.7$ ($55.8$)  \tabularnewline
\bottomrule
\end{tabular}
\par\end{center}
\caption{Github data: descriptive statistics.}
\label{si-github}
\end{table}

\end{document}